\documentclass{article}
\usepackage{amsmath,hyperref}\usepackage{color}
\usepackage{amsfonts,amssymb, bm, bbm}
\usepackage{latexsym}
\usepackage{hyperref}
\usepackage{graphicx}
\usepackage{amsthm}
\usepackage{color}
\usepackage{tikz}
\usepackage{xcolor}
\usepackage{MnSymbol}
\usepackage{subcaption}
\newcommand{\ds}{\displaystyle}

\newcommand{\ba}{\begin{array}}
\newcommand{\ea}{\end{array}}

\renewcommand{\r}{\right}
\newcommand{\be}{\begin{equation}}
\newcommand{\ee}{\end{equation}}
\newcommand{\eps}{\varepsilon}

\newcommand{\mc}{\mathcal}

\newcommand{\1}{\mathbbm{1}}
\newcommand{\E}{\mathbb{E}}

\newcommand{\R}{\mathbb{R}}

\renewcommand{\P}{\mathbb{P}}

\newcommand{\se}{\text{ if }}

\def\E{\mathbb{E}}
\def\R{\mathbb{R}}

\def\P{\mathbb{P}}

\def\qed{\hfill \vrule height 7pt width 7pt depth 0pt\medskip}
\def\beq{\begin{equation}}
\def\eeq{\end{equation}}

\newtheorem{theorem}{Theorem}
\newtheorem{definition}{Definition}
\newtheorem{proposition}{Proposition}
\newtheorem{lemma}{Lemma}
\newtheorem{corollary}{Corollary}

\begin{document}
\author{Giacomo Como\footnote{\textit{\ 
 Department of Mathematical Sciences G. Lagrange, Politecnico di
Torino.}},
Fabio Fagnani\footnote{\textit{\ 
 Department of Mathematical Sciences G. Lagrange, Politecnico di
Torino.}}, 
 Elisa Luciano\footnote{\textit{
Corresponding author: Elisa Luciano, ESOMAS Department and Collegio Carlo Alberto, Universit\`a di Torino. Email: elisa.luciano@unito.it}},\\
Alessandro Milazzo\footnote{\textit{\ 
ESOMAS Department, Universit\`a di
Torino.}}, 
Marco Scarsini\footnote{\textit{\ 
 Department of Economics and Finance, LUISS, Roma.}}
\footnote{The Authors thank Luigi Bocola, Alessandro Dovis and  Alireza Tazbah-Salehi for helpful early-stage comments, as well as Alessandro Lavia, Patrizia Semeraro and Francesca Zucchi for later comments. Luca Damonte provided research assistantship. Any errors are ours. Financing from PRIN PNRR P2022XT8C8-Supply chain disruptions, financial losses and their prevention – NextGenerationEU – CUP:  D53D23017960001 is gratefully acknowledged.}
}

\title{Rigidity and default  in production networks}
\date{\today}

\maketitle
\begin{abstract}
    This paper studies the transmission of productivity shocks in general equilibrium production networks, when firms in different sectors operate under informational rigidity and rely on external debt.  Rigidity breaks the Modigliani-Miller irrelevance of leverage and may generate default following shocks, even in equilibrium. 
    
    The economy consists of firms, banks, and consumers. Under proportional shock transmission, we prove that a unique Walrasian rigid equilibrium exists and provide explicit expressions for equilibrium quantities, prices, and interest rates.
    We show that, on the one hand, Hulten's theorem fails under rigidity, even without leverage. 
    On the other hand, we prove that welfare is smaller than in the first best if and only if both leverage and rigidity exist. The latter increase the total cost of debt  
    and have inflationary effects on the levered sectors, which propagate downstream, and shift consumption and labor upstream.  
    
    The occurrence of default depends solely on real shocks and the network structure, while  the magnitude of the losses depends also on  the connectedness of the economy and  the cost of debt of the connected sectors. We provide conditions for default cascades to occur and study two examples of default  propagation. \\\

\begin{classcode}
\end{classcode}
\end{abstract}


\section{Introduction}
Supply chain disruptions and other productivity shocks due, e.g., to pandemics, major natural disasters, wars or civil unrest, have recently surged and represent a serious threat to economic stability. The collapse of production we have witnessed during episodes of this type has generated firms' defaults  and  inability to pay debts back to banks and other financiers. Losses are no longer transmitted from the financial sector to the production sector; instead, they propagate in the opposite direction.

The consequences of real shocks would be less disruptive if producers could fully adapt to the shocks themselves. In particular, if they could perfectly anticipate all possible shocks and adjust their business decisions ---such as input of intermediate goods and labor--- as well as their financing decisions accordingly. They would simply operate at the first-best allocation under shocks and, in a competitive economy, still reach zero profits. As a consequence, they would still be able to pay back debts, under any shock realization, and would not harm their  creditors. This ideal scenario would produce greater welfare and would never lead to default.

In practice, what makes the interaction between the real and financial side inherently fragile is the presence of imperfect information.  In the real world, as emphasized by Pellet and Tahbaz-Salehi \cite{pellet2023rigid}, firms take decisions based on an imperfect knowledge of the possible shocks and,  as soon as a specific disruption occurs, they cannot  change their investment and production decisions at no cost and instantaneously. Production requires advance planning and organization of specific production lines that cannot be expanded or shrunk instantaneously.  Orders of intermediate goods, especially capital, take time to produce, ship, deliver, and disruptions of supply chains typically make this time longer, as we witnessed over the last years, due to wars and pandemics. Imperfect information therefore makes firms unable to take state-contingent decisions. Consistently with the literature, we call this inability \emph{rigidity}. 

Rigidity renders leverage relevant. By the time shocks realize, firms'  decisions are sunk and cannot be changed. Thus, firms  may suffer losses. They may even \emph{default} and, as a consequence, their financiers may be affected, because they recover less than due. This undermines the irrelevance of finance, or the separation of finance and operations, which holds in perfect capital markets à la Modigliani and Miller. 

The aim of this paper is to explore the consequences and transmission of real shocks in the presence of default. We examine whether and when these consequences are amplified by leverage, both before shocks occur and after. In the aftermath of real shocks occurrence, financiers may be hit. 

We start from  the observation that shocks propagate through firms' networks, and we distinguish between primitive and network-induced shocks. The former are the ones which originally hit a specific sector, the latter are transmitted from suppliers through input-output links. Sectors are made by homogeneous firms. They are also connected to financiers, or banks, who endogenously set the interest rate based on expected recovery. Default costs exist.

In this set-up, the paper explores the consequences in general equilibrium (GE) of the interaction between the decisions of the financiers and of the sectors, under rigidity.  We prove that, under an assumption of proportional shock transmission, a unique Walrasian rigid equilibrium exists and we explicitly provide its form.  

We show that Hulten's theorem (see Hulten \cite{hulten1978growth}) no longer holds in this setting, even in the absence of leverage: the effect of the shock in a sector is not related to its GDP share. As shocks take place, a sector with a small GDP share can have a disruptive effect on the overall welfare, as we witnessed over the last years. We prove that this is more likely to happen when the total shocks to a sector are large, in comparison with the shocks to other sectors, and its consequences are much more disruptive, the more levered the economy is.

Moreover, we prove that, due to rigidity, together with leverage, the first best is lost. Welfare in the rigid equilibrium is smaller than in the first best.  We describe the equilibrium outcome,  and how it compares to the first best, both in terms of prices and quantities. Prices increase and quantities decrease with the total cost of debt of a sector, which ---due to GE effects--- equals a weighted sum of the debt costs of the sector and its suppliers. The weights are given by the Leontief matrix, reflecting the lower importance of more distant sectors.

When a sector becomes levered, the prices of the goods it produces  increase and such  inflationary effects  propagate downstream, while optimal consumption and labor increase upstream. Thus, finance has both nominal and real consequences, even before shocks take place.

 We then focus  on default and sectors' losses. We comment on the conditions that lead to default and on the magnitude of the ensuing losses, which depend respectively on the real shocks and on the cost of debt, in both cases weighted using the Leontief matrix.
 
 Intuitively, default occurs for a sector when the shocks to neighboring nodes in the network are greater than its own: sectors are literally driven into default by their suppliers. So, default occurrence crucially depends on the network structure.
 
 When default occurs, the magnitude of losses depends on the overall cost of debt; this dependence is captured by a novel centrality measure, which we call \emph{discounted} Bonacich centrality, which weights upstream sectors by their debt-adjusted importance. For any given realization of the production shocks, the losses that a sector suffers are proportional to this centrality. As a result, we can study how amplification of losses in sectors that are downstream to a defaulting one occurs.
 
 Since default may be transmitted through the sectors' network, we establish conditions for default cascades to occur. We provide conditions  on the realization of shocks 
 under which either some or all of the subsequent nodes are affected.

The paper is structured as follows: Section \ref{lit} reviews the closest literature; Section \ref{sec:model} sets up our model; Section \ref{sec:eq} studies the existence and properties of the GE; Section \ref{sec: hulten} studies Domar's weights and Hulten's theorem failure; Section \ref{sec: unlevered} studies the subcase of an economy with informational frictions but no everage; Section \ref{sec: lev cons} studies the consequences of informational frictions and leverage  on real and nominal features, namely the demand for consumption, labor, as well as prices.  Section \ref{sec:def} is devoted to default occurrence, including default cascades,  and to the magnitude  of losses when default occurs; Section \ref{sec:ex} provides numerical examples and simulations, to illustrate how rigidity and  leverage affect the equilibrium outcome and the stochastic properties of sectors' profits and losses, depending on their network linkages. Section \ref{sec:conclu} summarizes and concludes. 
The proofs are collected in the Appendices.

\section{Related literature}\label{lit}

The background literature for this paper can be conveniently regrouped into four groups: the literature on network fragility, network models of the real side of the economy (supply networks or chains) without rigidity and debt, network models with rigidity and no debt, and network models without rigidity, but with debt. 

The fragility of modern, highly-interconnected economic systems due to wars, pandemics, social and geopolitical tensions is well known. The corresponding literature for supply chains, both empirical and theoretical, is rich (see Carvalho et al.\ \cite{carvalho2021supply}, Baqee and Farhi \cite{baqaee2019macroeconomic,baqaee2020productivity} and the review in Carvalho and Tazbah-Salehi \cite{carvalho2019production}). This literature concentrates on disruptions of supply chains, but, since firms are unlevered and they do not take rigid decisions, default cannot happen. An account of both the motivating evidence, the theoretical models of fragility description and some remedies  is in Elliott and Golub \cite{elliott2022networks}. To do so, the Authors study  supply and financial networks separately, and emphasize the role of failure costs in the latter: Networks can be either of the first or of the second type, and the two are not integrated.  They conclude with examples of policy actions to strengthen the resilience of supply or financial networks. We refer the reader not only to the evidence reported there, but also to the account of remedies for supply chains, which intervene on their structure (who supplies whom), and which we will not address here. Concerning our work, by linking firms in the supply networks to banks,  we will be able to suggest policy actions that act on the banks but aim at making the firms more resilient. With respect to direct changes in the supply network structure, policy interventions on the rate applied by banks to them or the amount of credit they can extend to specific nodes of the supply chain are relatively quick to impose.

The second bulk of literature we build on serves to model the real side of the economy before introducing rigidity and debt. In this respect, we use the approach typical of the GE literature on production networks, inspired by Long and Plosser \cite{long1983real} and Acemoglu et al.\ \cite{acemoglu2012network}. The corresponding literature  essentially aims at investigating the macroeconomic consequences of the microeconomic network structure of production and has three main features: first, the supply links are not relationship-based, and a specific firm is perfectly substitutable by one with the same production function. Second, the network (which supplies whom) is exogenous, and there is no attempt to optimize its links or modify its participants, for instance by changing the production function. Third, markets are competitive.  Recent surveys of this literature are Carvalho and Tazbah-Salehi \cite{carvalho2019production} and Baqee and Rubbo \cite{baqaee2023micro}, whilst a supply chain model which maintains the original approach but relaxes these three features is Acemoglu and Tazbah-Salehi \cite{acemoglu2024}.

Acemoglu at al.\ \cite{acemoglu2012network}, which is our main source of inspiration,  studies  a perfectly-competitive Cobb-Douglas economy  (production functions and consumer utility) in which  sectors adjust to shocks ex-ante by selecting  state-contingent production plans. 
As a consequence, input demands for intermediates, raw materials, and labor are state-dependent. Since decisions are made independently in each state, expected utility does not play a role and there is no rigidity. 
The absence of rigidity goes together with all equity financing. There is no debt financing, and obviously so: the way in which operations are financed is irrelevant, because without frictions such as rigidity the Modigliani-Miller separability holds. There is no debt and consequently no default.

The third important strand of literature that supports our paper is the rigid network one, in the absence of debt. The main mechanism embedded in the current paper is rigidity as in Pellet and Tahbaz-Salehi \cite{pellet2023rigid}. Also in their paper sectors are restricted in how effectively they can adjust their inputs to the economic environment (our states of the world), in the presence of  information frictions. This, as in our case, mimics the fact that in the real world production processes require advance planning and cannot be ramped up suddenly. There are lags in the delivery of inputs and they are relevant to shock transmission, aggregate output and welfare. 
Real rigidities together with informational frictions were examined in Angeletos et al.\ \cite{angeletos2016real}, but without allowing for an input-output network. 

In the economy of Pellet and Tahbaz-Salehi  \cite{pellet2023rigid} however there is no transmission of shocks from the real to the financial sector, because there is no external financing and no bank. 
Another key difference between our model and that of Pellet and Tahbaz-Salehi \cite{pellet2023rigid} is that we assume that all decisions are made by the sectors at time $t=0$, before observing the actual productivity shock realizations: as such, in our model, the sectors make their decisions based only on information available prior to the shock realization. In contrast, Pellet and Tahbaz-Salehi \cite{pellet2023rigid} allow part of the sectors' decisions to be made at time $t=0$ and part at time $t=1$: in particular, in their model, sectors choose their ``rigid intermediate input quantities'' at time $t=0$ based on partial information about the shocks, while at time $t=1$ they  ``optimally set their nominal prices and choose their labor and flexible intermediate input quantities to meet the realized demand'' upon fully observing the realizations of the shocks.

Rigidity arises from the need for advance production planning and the organization of specific production lines, as well as from the time required to produce, ship, and deliver intermediate goods, particularly when these are capital goods. Rigidity  is therefore different from the short- versus long-run adjustment to shocks embedded in Elliott and Jackson \cite{elliott2024supply}. In that paper indeed shocks come totally unpredicted, while in our model sectors know the prior distribution of shocks, which they update according to a fully informative or less-than-fully informative signal. We cover all possible information cases, not only the lack of information of Elliott and Jackson \cite{elliott2024supply}. As a consequence of the fact that shocks come as a surprise, sectors in Elliott and Jackson \cite{elliott2024supply} do not readjust the input-output decisions, and their production is simply constrained by the input shortages generated by the primitive shocks. As these shortages are transmitted through the network, in the short run  sectors seek to minimize shocks effects by maximizing revenues. In  the medium run prices adjust, and in the long run a new equilibrium, where profits are zero for every industry and Hulten's theorem holds, is reached. In our model, sectors use all available information to  form their optimal input decisions, even though they cannot make state-contingent orders, and therefore  reach the first best. As a consequence, they reach an equilibrium and trade inputs at shock-consistent prices before the actual shock occurs. When the shock hits, they can end up with profits or losses. Because of rigidity, Hulten's theorem fails. Therefore, the outcome of rigidity is very different from the mechanism in Elliott and Jackson \cite{elliott2024supply}, because the objective and therefore the assumptions differ.

A major consequence of rigidity in our paper is that Hulten's theorem fails to hold, and a single sector inability to produce the maximal or optimal quantities may have larger effects than the theorem commands. This is in line with anecdotal evidence from the recent years, where supply chain disruptions during COVID-19, wars or geopolitical tensions were responsible for aggregate losses well beyond the GDP share of the originally hit sector, as predicted by  Hulten's theorem. Semi-conductors are an example. 

The academic literature was aware of the limits of Hulten's theorem: major attempts to overcome them in GE are Baqee and Farhi \cite{baqaee2019macroeconomic,baqaee2020productivity}. The former paper preserves an efficient economy and studies non-linearities, which by definition are not captured by the differentials of the theorem. The second paper instead considers an inefficient economy, with arbitrary production functions and distortion wedges. We follow the second suggestion and explore how information frictions and default affect the impact of shocks, and may make it much larger than what the theorem says.

The fourth strand of literature we learned from is the one on production networks without rigidity but with debt. The need for external financing ---still without rigidity--- has  indeed been introduced by  Bigio and La'o \cite{bigio2020distortions}, with a model \`a la Acemoglu et al.\ \cite{acemoglu2012network}, through a limited enforcement (or cash-in-advance) mechanism. Inputs of intermediate goods and labor must be paid  before revenues are obtained, but after the occurrence of the shocks is revealed.  Sectors obtain loans from consumers themselves. The amount of the loan is state-dependent, so that there are never states of the world in which the cashed-in-advance outflow is greater than the optimal one, and cannot be refunded. Therefore, default never occurs. Our mechanism is different:  equity holders can receive external financing from debt holders (banks), who may ultimately obtain only a partial recovery if profits fall short. The banks require a debtor-specific interest rate, endogenously determined. In Bigio and La'o \cite{bigio2020distortions} the interest rate is exogenous.

The cash-in-advance mechanism, together with exogenous cost of debt, has been used also by Huremovic et al.\ \cite{huremovic2023production}, in an economy à la Acemoglou et al.\ \cite{acemoglu2012network} as far as the production technology (Cobb-Douglas) and the GE without rigidity are considered. 
We differ from this paper not only because we admit default and we set the interest rate endogenously, together with the other equilibium prices, but also  in the aim; while Huremovic et al.\ \cite{huremovic2023production} studies the consequences of changes in interest rates, and therefore the transmission of financial shocks to the real sector, when the two are related by credit, we aim at studying how production shocks, by generating cascades of defaults in the production network, affect the real as well as the financial sector. While in the aftermath of the Great Recession the question of contagion from the financial to the real side of the economy was indeed relevant, recent pandemic shocks on the economy, geopolitical tensions or war episodes command an understanding of whether and how much real shocks affect financiers. 

In the last sense our paper also differs from  Battiston et al.\ \cite{battiston2007credit}. Battiston et al.\ \cite{battiston2007credit} are interested in the demography of sectors and the correlation of output and bankruptcies. To do so, they investigate the effect of shocks to either prices  or produced quantities, which make firms unable to pay suppliers. Differently from Battiston et al.\ \cite{battiston2007credit}, we cannot have exogenous price shocks, since our prices are endogenous, and we examine the consequences of production shocks in GE, both on produced quantities and prices. Also, consistently with the Authors' aim, firms related by trade financing: they get credit from their suppliers. Only firms exist  and there are no banks.  Last but not least, there is no rigidity.\footnote{We consider the rich literature on financial networks {\it per se} not directly related to our work, because the banks in our model do not hedge the idiosyncratic risk from giving loans to different firms by engaging in any derivative or swap among themselves.}

\section{Economic model}\label{sec:model}
We consider a two-period  economy with a finite set $\mc V:=\{1,\ldots,n \}$ of \emph{sectors} or \emph{industries}. Every industry produces a single homogeneous good, which is partly used as an intermediate input for production by the other industries and partly consumed by a representative \emph{household}, who in turn supplies an aggregate constant unit of labor. 
Industries borrow from financiers ---that we call \emph{banks}---part of the capital needed to pay for their employed labor and consumption of intermediate goods. The profits of all industries and banks along with the total wage constitute the budget available to the representative household in order to pay for its own consumption of goods. 


\subsection{Sectors}\label{sec:industries}
We assume that sectors use constant returns-to-scale technologies to transform labor and intermediate inputs into their differentiated goods. Each sector is composed of a continuum of homogeneous firms that are randomly matched to the firms of the other sectors and to the representative household. Here, homogeneity amounts to that firms in the same sector have the same production function, and therefore they are perfectly substitutable.\footnote{\label{foot:hom }This is nested in the entry model of  Baqee and Farhi \cite{baqaee2020entry}.} 

We model production uncertainty in terms of a non-positive $n$-dimensional random vector $\eta$,\footnote{Throughout, we implicitly assume that $\eta$ and all random variables are defined on a common complete probability space $(\Omega,\mc F,\P)$.} whose entries $\eta_k\le0$ represent the \emph{primitive} log-productivity \emph{shocks} on the different sectors $k\in\mc V$.
Every log-productivity shock $\eta_k$ affects in the same way all firms within sector $k$.
The prior probability distribution of the log-productivity shock vector $\eta$ is assumed to be known to all actors in the economy. If the log-productivity shock vector $\eta$ is deterministic, then there is no uncertainty and the model is trivialized. For this reason, we assume that the log-productivity shock vector $\eta$ is not deterministic.\footnote{I.e., we assume that $\P(\eta=\eta^*)<1$ for every $\eta^*\in\R_-^{\mc V}.$}
In particular, this implies that there is always a positive probability that some primitive log-productivity shocks $\eta_k$ are strictly negative.
Throughout the paper, the adjective \emph{actual} refers to quantities (such as production, goods, assets, liabilities, etc.) that are determined after the realization of the shock $\eta$ and are therefore random.

We consider the Cobb-Douglas production functions 
\be\label{Cobb-Douglas-1}y^\eta_k=e^{\eta_k}\varsigma_kl_k^{\beta_k}\prod \limits_{j\in\mc V}(z^\eta_{jk})^{A_{jk}}\,,\ee
where for sector $k$:  $y_k^\eta$ is the actual (production) output; $z^\eta_{jk}$ is the actual used quantity of good $j$ and $l_k$ is the employed labor; $\beta_k\ge0$ is the labor share and $A_{jk}\ge0$ parametrizes the direct importance of the product $j$ in the sector $k$'s technology; 
and $\varsigma_k:=\beta_k^{-\beta_k}\prod_{j\in\mc V}A_{jk}^{-A_{jk}}$ is a positive normalization constant.\footnote{Throughout, we adopt the common convention $0^0=1$.} 
Our choice of the notation in Equation~\eqref{Cobb-Douglas-1} emphasizes the fact that actual output and actual amount of intermediate goods used in the production depend on the  primitive log-productivity shock vector $\eta$ ---with restrictions on the form of this dependence to be specified in the following---  hence they are themselves random. In the special case when the primitive log-productivity shock vector realization achieves its maximal value $\eta=0$, Equation~\eqref{Cobb-Douglas-1} reduces to $y^0_k=\varsigma_kl_k^{\beta_k}\prod\nolimits_{j\in\mc V}(z^0_{jk})^{A_{jk}}$. We shall refer to  $y_k^0$ and $z^0_{jk}$, respectively, as the \emph{maximal}\footnote{Sometimes referred to as ``steady state'' in the literature (e.g., in Pellet and Tahbaz-Salehi \cite{pellet2023rigid}). } (production) output and \emph{maximal}  quantity of good $j$ used by sector $k$. 
On the other hand, we are implicitly assuming that the quantities of labor $l_k$ employed by the different industries $k\in\mc V$ are not directly affected by the vector $\eta$ of the primitive log-productivity shocks. This is a standard assumption, capturing the idea that sectors must honor existing labor contracts and therefore continue to employ the contracted workforce despite production shocks.

Let $w$ denote the \emph{wage}, i.e., the  unit cost of the employed labor, and let $p_k$ be the unit costs of goods produced by sector $k\in\mc V$.  Therefore, the actual assets or revenues of each industry $k$ and  its actual liabilities due to the purchase of intermediate goods and the employment of labor are\footnote{Note that, importantly and differently from other network models, in our model productivity shocks affect both assets and liabilities.}
\be\label{assets}\mc A_k(\eta):=p_ky^\eta_k=p_ke^{\eta_k}\varsigma_kl_k^{\beta_k}\prod \limits_{j\in\mc V}(z^\eta_{jk})^{A_{jk}}\,,\ee
and
\be\label{liabilities}\mc L_k(\eta):=\sum\limits_{j\in\mc V}p_jz_{jk}^\eta+wl_k\,,\ee
respectively. 

Each industry $k\in\mc V$ borrows a given fraction $\theta_k$ of its actual liabilities $\mc L_k(\eta)$ from the banks at a rate $r_k\ge0$. The reader can interpret loans as credit lines. Differently from intermediate goods, that take time to produce and deliver and are rigid, loans can be issued or tailored to the firms' needs instantaneously. We assume that the leverage $\theta_k\in[0,1]$ is exogenously determined, while the interest rate $r_k$ is set by the banks themselves, as illustrated below in Section \ref{sec:bank}.\footnote{The leverage $\theta_k$ could be endogenized by introducing corporate taxes and the ensuing tax shield. We disregard the issue in this paper, because banks will in any case choose the product $r_k\theta_k$, namely adjust the interest rate to any level of leverage, and the equilibrium quantity will be characterized in terms of the product, not its components. See  Appendix \ref{sec:prop-zeta}.} We assume that, once the shocks are realized, labor and intermediate goods are paid first by the industries, while financiers receive only the minimum between what is due to them, i.e., the loan and its interests $(1+r_k)\theta_k\mc L_k(\eta)$, and what is available to the industry after labor and suppliers have been paid back, i.e.,\footnote{Throughout the paper, $[x]_+:=\max\{0,x\}$ stands for the positive part of a real number $x$. }$[\mc A_k(\eta)-(1-\theta_k)\mc L_k(\eta)]_+$. The remainder of what is due to the banks, i.e.,
\be\label{default-costs}\mc D_k(\eta):=\left[(1+r_k)\theta_k\mc L_k(\eta) -[\mc A_k(\eta)-(1-\theta_k)\mc L_k(\eta)]_+\right]_+ \,,\ee
is the \emph{default cost} that industry $k$ pays ---for instance for lawyers and labor costs for dismantling production lines--- but the banks do not collect.  The actual profit\footnote{For the sake of simplicity, we keep using the term ``profit'' for $\pi_k(\eta)$ instead of referring to it as profit or loss depending on its sign.} of industry $k$ is then  
\be\label{profit}\pi_k(\eta):=\mc A_k(\eta)-(1+r_k\theta_k)\mc L_k(\eta)
\,.\ee
As we show later in the paper, in the presence of informational frictions and rigidity, the actual profits $\pi_k(\eta)$  can be negative as well as positive in equilibrium, even though they are all zero in expectation. As a consequence, in our model, sectors may \emph{default} in equilibrium.

We model informational frictions and rigidity in the economy by assuming that all industries (as well as the representative household and banks, as detailed later in Sections \ref{sec:household} and  \ref{sec:bank}, respectively) make decisions at time $0$ based only on partial information on the actual realization of the log-productivity shock vector $\eta$. The actual realization becomes fully known only at time $1$, when all  decisions made at time $0$ and the corresponding payments are executed. 
In particular, we assume that, at time $0$, all economic actors observe  the same \emph{public signal} 
$\varphi(\eta)$, 
where $\varphi:\R_-^n\to\mc S$ is a map from the space of the log-productivity shock vector realizations to a non-empty set of public signals $\mc S$.\footnote{Throughout, both the signal set $\mc S$ and the map $\varphi$ are assumed to be measurable.} 
At time $0$, each industry $k\in\mc V$ decides the vector of maximal quantities of intermediate goods $(z^0_{jk})_{j\in\mc V}$ to order and the labor $l_k$ to employ, with the objective of maximizing its conditional expected profit 
\be\label{cond-exp-profit}\E[\pi_k(\eta)|\varphi(\eta)]\,,\ee
given the public signal. 

Observe that, as an extreme special case where $\mc S=\R_-^n$ and $\varphi$ is the identity map (i.e., $\varphi(\eta)=\eta$), we recover the case where industries make decisions informed by full observation of the realization of the primitive log-productivity shock vector $\eta$ and the rigidity vanishes. This happens when industries can adapt their decisions to the realization of the shock. In this case, the decision on intermediate goods of sector $k$ is not a vector of constants, but a  vector of random variables. Correspondingly, sector $k$ chooses a labor quantity per state, and not a single optimal labor quantity independent of the realization of the log-productivity shock vector $\eta$. 

At the other end of the spectrum, when $\varphi\equiv C$ is a constant map so that the public signal  $\varphi(\eta)=C$ carries no information about the realization of the primitive log-productivity shock vector $\eta$, industries make decisions based only on their common prior probability distribution of $\eta$. In this case, rigidity in the choice of intermediate quantities and labor {\it a fortiori} applies, and sector $k$ chooses an unconditional quantity of intermediate goods and labor. 

The signal function $\varphi$ allows us to parametrize the informational rigidity in the economy between the full information case\footnote{In fact, we obtain the full information case whenever $\eta$ is measurable with respect to $\varphi(\eta)$. Throughout this article, for the sake of simplicity, we refer to $\varphi(\eta)=\eta$ as the full information case.} $\varphi(\eta)=\eta$ and the no information case $\varphi(\eta)=C$. More precisely, we say that a signal function $\varphi:\R_-^{\mc V}\to\mc S$ is at least as informative as another signal function $\tilde\varphi:\R_-^{\mc V}\to\tilde{\mc S}$ if there exists a measurable map $\phi:{\mc S}\to\tilde{\mc S}$ such that $\tilde{\varphi}(\eta)=\phi(\varphi(\eta))$. This defines a partial ordering among signal functions, whereby constant signals $\varphi(\eta)=C$ are the least informative (no information) and the identity signal $\varphi(\eta)=\eta$ is the most informative. Notice that all signal functions would be equally informative if the log-productivity shock vector $\eta$ were deterministic, a degenerate case that we excluded from our treatment. 

In order to relate the maximal quantities of intermediate goods  ordered by industries at time $0$ (which, generally, do not depend directly on the actual realization of the log-productivity shock vector $\eta$ but only on its conditional distribution given the observed public signal $\varphi$) to the corresponding actual quantities  purchased and used by industries at time $1$ (which do depend on the realization of the log-productivity shock vector $\eta$), we consider the following \emph{proportional rationing} rule:
\be\label{proportional-rationing-1} z^{\eta}_{kj}/z^{0}_{kj}=y^{\eta}_k/y^{0}_k\,,\qquad \forall j,k\in\mc V\text{ s.t. } z^{0}_{kj}>0\,,\ y^0_k>0\,.\ee
Equation \eqref{proportional-rationing-1} entails that, if the log-productivity shock vector $\eta$ reduces the output of sector $k$ with respect to its maximal value $y^0_k$ (the one that would have been observed in the ideal scenario $\eta=0$) by some percentage, then  at time $1$ every industry $j$ that is a customer of sector $k$ receives a quantity $z_{kj}^{\eta}$ of that good that is decreased by the same percentage with respect to the maximal value $z_{kj}^0$ that had been ordered at time $0$. This proportional rationing rule is equivalent to assuming that all  customers have the same priority. 

\subsection{Banks }\label{sec:bank}

 Banks operate in perfect competition and free entry. There exists a representative bank that finances sector $k$, with a loan that covers a fraction  $\theta_k$ of  its actual liabilities.
The bank is assumed  to choose the interest rate $r_k$, while the loan amount is taken as given.

In order to describe the profits of the banks, 
first observe that their credit towards sector $k$ is $(1+r_k)\theta_k\mc L_k(\eta)$ 
and that the amount they will in fact be able to recover from sector $k$ is given by 
$$\mc R_k(\eta):=(1+r_k)\theta_k\mc L_k(\eta)-\mc D_k(\eta)\,.$$
It follows that the banks' profits\footnote{As for the industries, we refer to $\mc I_k(\eta)$ as the bank's profit, with the understanding that, when negative, they correspond to losses.} are 
\be\label{Bank-profit}
\mathcal I_k(\eta)
:=\mc R_k(\eta)-\theta_k\mc L_k(\eta)
=r_k\theta_k\mc L_k(\eta)-\mc D_k(\eta)
\,.
\ee

As with industries and the representative household, we assume that banks financing sector $k$ choose the interest rate $r_k$ at time $0$ based on the partial information provided by the observation of the public signal $\varphi(\eta)$. 
Banks are assumed to be risk-neutral and competitive,  and therefore set the interest rate $r_k$ so that their conditional expected profit from sector $k$ given $\varphi(\eta)$ is 
null, i.e., 
\be\label{bank-optimality}\E[\mc I_k(\eta)|\varphi(\eta)]
=0\,.\ee
As a consequence of Proposition \ref{prop-zeta}, that will be formally stated in Section \ref{sec:eq} and proved in Appendix \ref{sec:prop-zeta}, we will get that, for $\theta_k >0$, the solution $r_k$ to the previous equation exists and is unique. Such $r_k$ will be shown to be zero in the special case of perfect information, while it is typically positive when information is not full.\footnote{\label{theta0} For $\theta_k =0$, i.e., when the banks do not provide financing to a sector, equation  \eqref{bank-optimality} is satisfied for all real values of the interest rate. In what follows, in order to include the case without leverage in our comments and to be able to discuss the solutions for every $\theta_k$ in $[0,1]$, we take as solution for $\theta_k=0$ the limit of the solutions of \eqref{bank-optimality} as $\theta_k \downarrow0$. In Section \ref{sec:eq} below we show that it is zero.}

\subsection{Representative household}\label{sec:household}
The representative household consumes quantities $c_k^{\eta}$, for $k\in\mc V$, of the produced goods. 
We assume that she has a constant returns-to-scale Cobb-Douglas utility function  
\be\label{Cobb-Douglas-2}U(c^{{\eta}})=\chi\prod_{k\in\mc V}(c_k^{\eta})^{\gamma_k}\,,\ee
where $\gamma_k\ge0$ is the \emph{consumer preference} weight for the good produced by the industries $k\in\mc V$ and $\chi:=\prod_{k\in\mc V}\gamma_k^{-\gamma_k}$ is a positive normalization constant. 

The representative household receives labor wage as worker, obtains profits as owner of both industries and banks, and collects all default  costs.\footnote{Note that firms' profits include default costs, which are therefore paid to consumers. They are sunk to firms but not to the economy.} To compute the representative household's endowment, first notice that Equations \eqref{profit} and \eqref{Bank-profit} yield $$\pi_k(\eta)+\mc I_k(\eta)+\mc D_k(\eta)=\mc A_k(\eta)-(1+r_k\theta_k)\mc L_k(\eta)+r_k\theta_k\mc L_k(\eta)=\mc A_k(\eta)-\mc L_k(\eta)\,,$$ for every sector $k\in\mc V$. Then, also using the unitary supply of labor assumption, it follows that the representative household's endowment amounts to
\be\label{budget}
\mc E(\eta)
:=\ds w\sum_{k\in\mc V}l_k+\sum\limits_{k\in\mc V}\pi_k(\eta)+\sum\limits_{k\in\mc V}\mc I_k(\eta)+\sum\limits_{k\in\mc V}\mc D_k(\eta)
=\ds w+\sum_{k\in\mc V}(\mc A_k(\eta)-\mc L_k(\eta))\,.
\eeq
Analogously to the industries, the representative household makes decisions at time $0$ on the maximal quantities $c^{0}_k$ of the different goods to order, 
aiming at maximizing the conditional expected value of her utility
\be\label{cond-exp-utility}\E[U(c^{\eta})|\varphi(\eta)]\,,\ee
given the public signal $\varphi(\eta)$, 
under the budget constraint  
\beq\label{budget-constraint} \sum\limits_{k\in\mc V} c^{\eta}_kp_k\leq\mc E(\eta).\ee
Again, similarly to the mechanism described in Section \ref{sec:industries} for the industries, the actual consumptions of the representative household are determined by the following proportional rationing rule
\be\label{proportional-rationing-2} c^{\eta}_{k}/c^{0}_{k}=y^{\eta}_k/y^{0}_k\,,\qquad \forall k\in\mc V\text{ s.t. } c^{0}_{k}>0\,,\ y^0_k>0\,.\ee


\subsection{Equilibrium definition}\label{eqs}
Throughout the paper, we consider Cobb-Douglas economies parametrized by the  matrix $A=(A_{jk})_{j,k\in\mc V}$, the non-negative vector $\beta=(\beta_k)_{k\in\mc V}$ of the intermediate product importance parameters of the labor shares, and the non-negative vector $\gamma=(\gamma_k)_{k\in\mc V}$ of the final consumer preferences. 
Our constant returns-to-scale assumptions amount to
\be\label{normalization-1} \sum_{j\in\mc V}A_{jk}+\beta_k=1\,,\qquad\forall k\in\mc V\,,\ee
and
\be\label{normalization-2} \sum_{k\in\mc V}\gamma_k=1\,.\ee 
Hence, we can identify a Cobb-Douglas economy with constant returns-to-scale by the pair $(A,\gamma)$, where $A$ is a non-negative square with column sums not exceeding $1$, and $\gamma$ is a non-negative vector whose entries sum up to $1$.

We propose the following definition of \emph{rigid} general equilibrium (GE).

\begin{definition}\label{def:Walras-bank}
Consider a constant returns-to-scale Cobb-Douglas economy $(A,\gamma)$ and leverage vector $\theta$. 
Let $\eta$ be a primitive log-productivity shock vector and  let $\varphi(\eta)$ be a public signal. 
Then, a \emph{$\varphi$-rigid Walrasian equilibrium} is a tuple $$(y^0,z^0,l,c^0,r,p,w)$$ such that: 
\begin{enumerate} 
\item[(i)] the employed labor $l_k$ and the maximal quantities of intermediate goods $(z^0_{jk})_{j}$ ordered at time $0$ by each industry $k\in\mc V$,  maximize its conditional expected profit $\E[\pi_k(\eta)|\varphi(\eta)]$;
\item[(ii)] the household's maximal consumption vector $c^0$ ordered at time $0$ maximizes the conditional expected consumer utility $\E[U(c^{\eta})|\varphi(\eta)]$ under the shock-dependent budget constraint \eqref{budget-constraint}; 
\item[(iii)] for every $k\in\mc V$, the banks  financing sector $k$ make zero conditional expected profit $\E[\mc I_k(\eta)|\varphi(\eta)]=0$; 
\item[(iv)] the markets for goods clear
\be\label{market-clearing}y_k^\eta=\sum\limits_{j\in\mc V}z^\eta_{kj}+c^\eta_k\,,\qquad\forall k\in\mc V\,;\ee
\item[(v)] the market for labor clears 
\be\label{labor-clearing}\sum\limits_{k\in\mc V}l_k=1\,;\ee
\end{enumerate}
\end{definition}

It is worth pointing out a few key features in our model and in the notion of rigid Walrasian equilibrium proposed above. First, recall that the decision variables $z_{jk}^0$ and $l_k$ (for industries), $c_0^k$ (for the representative household), and $r_k$ (for banks), are measurable with respect to the public signal $\varphi(\eta)$ that is observable by all economic actors at time $0$. Hence, they depend on the log-productivity shock vector $\eta$ only through $\varphi(\eta)$. In particular, in the special case of no information, these quantities do not depend at all on the realization of the log-productivity vector $\eta$. In fact, the same holds true for the prices $p_k$ and the wage $w$, as well as for the conditional expectations of the industries' profits $\E[\pi_k(\eta)|\varphi(\eta)]$, of the consumer utility $\E[U(c^{\eta})|\varphi(\eta)]$, and of the banks' profits $\E[\mc I_k(\eta)|\varphi(\eta)]$, which appear in items (i), (ii), and (iii), respectively, of Definition \ref{def:Walras-bank}. In contrast, actual industry productions $y_k^{\eta}$, intermediate purchases quantities $z_{jk}^{\eta}$, and household consumptions $c_k^{\eta}$ do depend on the actual realization of the log-productivity shock vector $\eta$, and this dependence will be made explicit in Proposition \ref{prop:actual}. Hence, in particular, the market for goods clearing constraints \eqref{market-clearing} as well as the consumer budget constraint \eqref{budget-constraint} are contingent on the realization of the  log-productivity shock vector $\eta$ (whereas the market for labor  clearing constraint \eqref{labor-clearing} only involves quantities that are measurable with respect to the public signal $\varphi(\eta)$). 

Finally, notice that, as is common in these models, $\varphi$-rigid equilibria can be uniquely identified only up to a numeraire. In fact, if $(y^0,z^0,l,c^0,r,p,w)$ is a $\varphi$-rigid equilibrium, then so is $(y^0,z^0,l,c^0,r,\kappa p,\kappa w)$, for every positive-real-valued $\varphi(\eta)$-measurable random variable $\kappa$. For this reason, we shall refer to a $\varphi$-rigid equilibrium as unique if all $\varphi$-rigid equilibria have the same $y^0$, $z^0$, $l$, $c^0$, $r$, and $p/w$.

\section{Rigid equilibrium analysis\label{sec:eq}} 

\subsection{Some auxiliary quantities of interest}
Observe that the constant returns-to-scale assumption \eqref{normalization-1} for industries' production functions implies that no column sum of the matrix $A$ exceeds one. Hence, $A$ is a column-substochastic matrix, so that its spectral radius, i.e., the largest magnitude of its eigenvalues, does not exceed $1$. Throughout, we shall assume that the spectral radius of the matrix $A$ is strictly less than one.\footnote{ This is a mild assumption as it is equivalent to saying that for every industry $k$, there exists a path $i_0,i_1,\dots, i_l=k$ such that $\beta_{i_0}>0$ and $A_{i_{s-1}i_{s}}>0$ for every $s=1,\dots , l$. In particular, this holds true if $\beta_k>0$ for all $k$ (i.e., if every sector employs some labor), but it is in fact a less restrictive condition (it is sufficient that every sector is connected through an arbitrary number of hops to at least a sector that employs labor).} As a consequence, the commonly known economy's \emph{Leontief inverse} $L:=(I-A')^{-1}$ is a well-defined matrix that can be expressed as the series of the powers of $A'$, \be\label{Leontief-expansion}L=(I-A')^{-1}=\sum_{h=0}^{+\infty}(A')^h\,.\ee Hence, in particular, the economy's Leontief inverse matrix $L$ has all non-negative entries. In particular, the $(j,k)$'th entry $L_{jk}=\delta^k_j+A_{kj}+(A^2)_{kj}+\ldots$ of the Leontief matrix measures the total (direct and indirect) importance of product $k$ in sector $j$'s technology. 
 We shall refer to an industry $j$ as a \emph{supplier} of another industry $k$ if either $j\ne k$ and $L_{kj}>0$ or $j=k$ and $L_{kk}>1$; notice that this definition includes both \emph{direct} suppliers, i.e., sectors $j$ such that $A_{jk}>0$,  and \emph{indirect} suppliers, i.e., sectors $j$ such that $A_{jk}=0$ and $(A^{h})_{jk}>0$ for some $h\ge2$. Symmetrically, we shall refer to an industry $j$ as a \emph{customer} of another industry $k$ if either $j\ne k$ and $L_{jk}>0$ or $j=k$ and $L_{kk}>1$, thus including both direct customers, i.e., sectors $j$ such that $A_{kj}>0$,  and indirect customers, i.e., sectors $j$ such that $A_{kj}=0$ and $(A^{h})_{kj}>0$ for some $h\ge2$. 

It is standard to refer to the $k$'th entry \be\label{Bonacich-def}v_k^0=\sum_{j\in\mc V}\gamma_jL_{jk}\,,\ee 
of the vector $v^0=L\gamma$, as the  \emph{Bonacich centrality measure} of the sector $k$ in a constant returns-to-scale Cobb-Douglas economy $(A,\gamma)$.  Observe that the Bonacich centrality $v^0_k$ of a sector $k$ is positive if and only if either $\gamma_k>0$ (i.e., sector $k$ is a direct supplier of the final consumer) or $\gamma_k=0$ and $k$ is a supplier of some sector $j$ with $\gamma_j>0$ (i.e., sector $k$ is an indirect supplier of the final consumer).

A primary role in the propagation of the shocks in the economy is played by the random vector of the \emph{total} log-productivity shocks, defined as \be\label{rho}\rho:=L\eta=(I-A')^{-1}\eta\,.\ee
Observe that Equations \eqref{rho} and \eqref{Leontief-expansion} imply that the total log-productivity shock in an industry $k\in\mc V$, 
 \be\label{rho-def}\rho_k=\sum_{j\in\mc V}L_{kj}\eta_j=\eta_k+\sum_{j\in\mc V}\eta_jA_{jk}+\sum_{i\in\mc V}\sum_{j\in\mc V}\eta_iA_{ij}A_{jk}+\ldots\,,\ee 
 is the aggregate of the primitive log-productivity shock $\eta_k$ on that industry, the sum $\sum_{j\in\mc V}\eta_jA_{jk}$ of the primitive log-productivity shocks $\eta_j$ in the direct suppliers $j$ each weighted by the importance of product $j$ in sector $k$'s technology, the weighted sum $\sum_{i\in\mc V}\sum_{j\in\mc V}\eta_iA_{ij}A_{kj}$ of the primitive log-productivity shocks $\eta_i$ in the industries $i$ supplying sector $k$'s direct  suppliers, and so on, accounting for all indirect effects of the primary log-productivity shocks in the industries upstream to sector $k$ in the production network. The following result illustrates the role of the log-productivity shocks in connection with the proportional rationing rule \eqref{proportional-rationing-1}: its proof is presented in Appendix \ref{sec:proof-proposition1}. 

\begin{proposition}\label{prop:actual}
Consider a constant returns-to-scale Cobb-Douglas economy $(A,\gamma)$. Let $l$ be a labor vector satisfying the clearing condition \eqref{labor-clearing} and let $\eta$ be a primitive log-productivity shock. Then: 
\begin{enumerate}
\item[(i)] if the Cobb-Douglas production relation \eqref{Cobb-Douglas-1}, the proportional rationing rule \eqref{proportional-rationing-1} and the market for goods clearing condition \eqref{market-clearing} are satisfied for every realization of the log-productivity shock vector $\eta$, then \be\label{actual-quantities} y^{\eta}_k=e^{\rho_k}y^{0}_k\,,\qquad 
z^{\eta}_{jk}=e^{\rho_j}z^{0}_{jk}\,,\qquad
c^{\eta}_k=e^{\rho_k}c^{0}_k\,,\ee 
is satisfied for every sectors $k$ and $j\in\mc V$, 
where $\rho_k$ is the network-induced log-productivity shock, as defined in \eqref{rho-def};
\item[(ii)] if the actual productions, intermediate quantities, and household consumptions satisfy  Equation \eqref{actual-quantities}  for every sectors $k$ and $j\in\mc V$,  where $\rho_k$ is the network-induced log-productivity shock, as defined in \eqref{rho-def}, then the Cobb-Douglas production relation \eqref{Cobb-Douglas-1} and the market for goods clearing condition \eqref{market-clearing} are satisfied for every realization of the log-productivity shock vector $\eta$ if and only if they are satisfied for $\eta=0$. 
\end{enumerate}

\end{proposition}


It is now convenient to introduce the \emph{information-normalized total shock} on sector $k\in\mc V$, 
\be\label{tau-def}\tau_k:=\frac{e^{\rho_k}}{\E[e^{\rho_k}|\varphi(\eta)]}\,,\ee
i.e., the exponential of the sector's total log-productivity shock normalized by its own conditional expected value given the public signal $\varphi(\eta)$ observable by all the actors in the economy at time $0$. Let also
\be\label{epsk}\eps_k:=\sum_{j\in\mc V}\tau_jA_{jk}+\beta_k\,,\ee
be the \emph{information-normalized total shock to the suppliers}  of a sector $k\in\mc V$, i.e., the sum of information-normalized total shocks to all suppliers of sector $k$, weighted by their importance, including the shock to labor (which is equal to one, since labor does not change its productivity) multiplied by the labor share. 
In the special case of full information (i.e., when $\varphi(\eta)=\eta$) for every sector $k\in\mc V$ we have that $\E[e^{\rho_k}|\varphi(\eta)]=e^{\rho_k}$, so that $\tau_k=1$ and $\eps_k=\sum_jA_{jk}+\beta_k=1$, thanks to the constant returns-to-scale assumption \eqref{normalization-1} for the industries' production functions: hence, in this case, $\tau_k=\eps_k=1$ are deterministic. 
On the other hand, in the general case of partial information,  $\tau_k$ and $\eps_k$ are random variables with conditional expected values 
\be\label{Etauk=1}\E[\tau_k|\varphi(\eta)]=1\,,\qquad \forall k\in\mc V\,,\ee
and 
\be\label{Eepsk=1}\E[\eps_k|\varphi(\eta)]=\sum_{j\in\mc V}A_{jk}+\beta_k=1\,,\qquad \forall k\in\mc V\,,\ee
respectively, where the latter follows again from the constant returns-to-scale assumption \eqref{normalization-1} for the industries' production functions. The randomness of $\tau_k$ reflects the remaining uncertainty on the total shock $e^{\rho_k}$ on sector $k$ upon the observation of the public signal $\varphi(\eta)$ at time $0$ by all actors in the economy. Similarly, the randomness of $\eps_k$ reflects the remaining uncertainty on the aggregate of the total shocks  $e^{\rho_j}$ on the suppliers $j$ of sector $k$ weighted by their importance $A_{jk}$ in the production technology of industry $k$.

We now present the following result, which is proved in  Appendix \ref{sec:prop-zeta}. There we demonstrate that the zero conditional expected profit assumption \eqref{bank-optimality} can be written in terms of the transform $\zeta_k$ of the product $r_k \theta_k$, where $1+r_k \theta_k= e^{\zeta_k}$. For this reason, we shall refer to $\zeta_k$ as the \emph{primitive cost of debt} on sector $k$.

\begin{proposition}\label{prop-zeta}
For every sector $k\in\mc V$ and leverage value $\theta_k\in[0,1]$, the equation \be\label{distortion-def}
\E\left[\left[e^{\zeta_k}\tau_k\right]^{e^{\zeta_k}\eps_k}_{(1-\theta_k)\eps_k}
\big|\varphi(\eta)\right]=1\,,
\ee
 admits a unique non-negative solution $\zeta_k=\zeta_k(\theta_k,\varphi(\eta))$. Such solution is non-decreasing as a function of $\theta_k$, with \be\label{zeta-min}\zeta_k(0,\varphi(\eta))=0\,,\ee
and 
\be\label{zeta-max}\zeta_k(1,\varphi(\eta))=-\log\E[\min\{\eps_k,\tau_k\}|\varphi(\eta)]\,.\ee 
Furthermore, in the special case 
of full information (i.e., when $\varphi(\eta)=\eta$), we have 
\be\label{zeta=0}\zeta_k(\theta_k,\eta)=0\,,\qquad\forall \theta_k\in[0,1]\,.\ee 
\end{proposition}
 Proposition \ref{prop-zeta} ensures that the primitive cost of debt $\zeta_k$ for a sector $k$ is well defined and it is a monotone non-decreasing function of the leverage $\theta_k$, ranging from a minimum value $\zeta_k=0$ when $\theta_k=0$ 
(i.e., no cost of debt, when there is no leverage), to a maximum\footnote{Note that $\E[\min\{\tau_k,\eps_k\}|\varphi(\eta)]\leq \E[\tau_k,|\varphi(\eta)]\!=\!1$ so that $\zeta_k\!=\!-\log\E[\min\{\tau_k,\eps_k\}|\varphi(\eta)]\geq 0$.} value 
$\zeta_k=-\log\E[\min\{\tau_k,\eps_k\}|\varphi(\eta)]$ 
achieved  when $\theta_k=1$, i.e., when sector $k$ borrows the entirety of its liabilities.
Proposition \ref{prop-zeta} also states that, in the special case of full information $\varphi(\eta)=\eta$, the primitive cost of debt is $\zeta_k=0$ for every sector $k$ regardless of its leverage $\theta_k$, due to the fact that the information-normalized total shocks are deterministic, i.e., $\tau_k=1$ and $\eps_k=1$.
In the general case of partial information, the costs of debt $\zeta_k$'s are non-negative real-valued random variables that are measurable with respect to the public signal $\varphi(\eta)$. Hence, the cost of debt can be positive only when information is not full, and only on those sectors that borrow a positive fraction of their liabilities. 

We also denote by 
\be\label{xi-def}\xi_k:=\sum_{j\in\mc V}L_{kj}\zeta_j=\zeta_k+\sum_{j\in\mc V}\zeta_jA_{jk}+\sum_{i\in\mc V}\sum_{j\in\mc V}\zeta_iA_{ij}A_{jk}+\ldots\,,\ee
the \emph{total cost of debt} of  sector $k\in\mc V$: 
this is the aggregate of the primitive cost of debt $\zeta_k$ on industry $k$, the sum $\sum_{j\in\mc V}\zeta_jA_{jk}$ of the costs of debt $\zeta_j$ of the supplying industries $j$ each weighted by the importance of product $j$ in sector $k$'s technology, the weighted sum $\sum_{i\in\mc V}\sum_{j\in\mc V}\zeta_iA_{ij}A_{kj}$ of the primitive costs of debt $\zeta_i$ in the industries $i$ supplying sector's $k$ suppliers, and so on, accounting for all indirect effects of the costs of debt in the industries upstream to sector $k$ in the production network. The sum  is weighted using the Leontief matrix, since far sectors are less important in the downstream output.
The discount factors corresponding to the primitive and total cost of debt are $e^{-\zeta_k}$ and$ e^{-\xi_k}$, respectively.

Define the \emph{$\zeta$-discounted Leontief matrix} \be\label{distorted-Leontief}L^{\zeta}:=(I- e^{-[\zeta]}A')^{-1}\,,\ee where $ e^{-[\zeta]}$ is the diagonal matrix with the discount factors $e^{-\zeta_k}$ for the different sectors $k\in\mc V$ on the main diagonal. Observe that the discounted Leontief matrix is a well defined non-negative square matrix that can be expressed as the limit of the geometric series
$L^{\zeta}=(I- e^{-[\zeta]}A')^{-1}=\sum_{h=0}^{+\infty}(e^{-[\zeta]}A')^h\,.$
We introduce also the normalized \emph{discounted network centrality}, which, for each sector $k\in\mc V$, extends the standard Bonacich centrality measure of connectedness as follows:
\begin{equation}
\label{centrality-def}\begin{array}{rcl}v_k^{\zeta}&:=&\ds\frac{1}{\psi(\zeta)}\ds\sum_{j\in\mc V}L_{jk}^{\zeta}\gamma_j\\
&=&\ds\frac1{\psi(\zeta)}\left(\gamma_k+\sum_{j\in\mc V}A_{kj}e^{-\zeta_j}\gamma_j+\sum_{i\in\mc V}\sum_{j\in\mc V}A_{kj}A_{ji}e^{-\zeta_j-\zeta_i}\gamma_i+\ldots\right)
\,,\end{array}
\end{equation}
where 
\be\label{psi}\psi(\zeta):=\sum_{j\in\mc V}\sum_{k\in\mc V}\gamma_jL_{jk}^{\zeta}\beta_ke^{-\zeta_k}\,,\ee
is a normalization factor that ensures that   $\sum_kv_k^{\zeta}\beta_ke^{-\zeta_k}=1$. 

In the special case when  $\zeta=0$, i.e., when the cost of debt is null on every sector (e.g., when there is full information, or when the leverage is zero on all sectors) the discounted Leontief matrix reduces to the standard one $L^0=L=(I-A')^{-1}$ and the constant returns-to-scale assumption $\beta_k=1-\sum_{j}A_{jk}=((I-A')\mathbf 1)_k$ implies that 
$L\beta=L(I-A')\mathbf 1=\mathbf 1$, so that the normalization factor reduces to 
$$\psi(0)=\gamma'L\beta=\mathbf 1'\gamma=1\,.$$ Therefore, when $\zeta=0$, the normalized discounted network centrality of each sector $k$ reduces to the 
Bonacich centrality measure of Equation \eqref{Bonacich-def}. 


The following result, proved in Appendix \ref{sec:proof-lemma-monotonicity-zeta}, establishes some fundamental monotonicity properties for the vector $\zeta$ of the costs of debt.  Four of them areof direct economic interpretation and interest, while the last two will be used in the proof of Proposition \ref{prop:lev-effect-qp}.
\begin{lemma}\label{lemma:monotonicity-zeta}
Let $\zeta$ be the vector of the primitive costs of debt. Then: 
\begin{enumerate}
\item[(i)] every entry $L_{jk}^{\zeta}$ of the discounted Leontief matrix is monotone non-increasing in $\zeta$; 
\item[(ii)] the normalization factor $\psi(\zeta)$ is monotone non-increasing in $\zeta$; 
\item[(iii)] the (non-normalized) discounted network centrality $\psi(\zeta)v^{\zeta}_k$ of every sector $k\in\mc V$ is monotone non-increasing in $\zeta$. 
\item[(iv)] if sector $k$ is not a supplier of sector $o$, then $L_{jk}^{\zeta}$ is constant in $\zeta_o$ for every sector $j\in\mc V$, and $v^{\zeta}_k$ is monotone non-decreasing in $\zeta_o$;
\item[(v)]  the normalization factor $\psi(\zeta)$ is strictly decreasing in $\zeta_o$ for every sector $o$ such that $v^0_o>0$; 
\item[(vi)] if sector $k$ is not a supplier of sector $o$ and $v^0_o>0$, then $v^{\zeta}_k$ is strictly increasing in $\zeta_o$.  
\end{enumerate}
\end{lemma}

\subsection{Rigid equilibrium characterization} 
The following result, proved in Appendix \ref{sec:proof-theorem1}, ensures the existence and uniqueness of the rigid Walrasian equilibrium and characterizes it. 
\begin{theorem}\label{theo:Walrasian} 
Consider a constant returns-to-scale Cobb-Douglas economy $(A,\gamma)$ and leverage vector $\theta$. 
Let $\eta$ be a primitive log-productivity shock vector and  let $\varphi(\eta)$ be a public signal. 
Then, there exists a unique $\varphi$-rigid Walrasian equilibrium $(y^0,z^0,l,c^0,r, p,w)$.  Moreover, for every sector $k\in\mc V$, let $\zeta_k$, $\xi_k$, and $v_k^{\zeta}$ be the primitive cost of debt, the total cost of debt, and the discounted centrality, as defined in Equations \eqref{distortion-def}, \eqref{xi-def}, and \eqref{centrality-def}, respectively. Then, for every two sectors $j,k\in\mc V$, we have the following:  
\begin{itemize}    
\item\textbf{maximal productions:}
\be\label{max-y} y^0_k 
           =v_k^{\zeta}e^{-\xi_k}\,;   \ee   
    \item\textbf{maximal intermediate quantities:}
     \be\label{max-z} z^0_{jk}
         =v_k^{\zeta} A_{jk}e^{-\zeta_k-\xi_j};\ee
    \item\textbf{employed labor:}
\be\label{labor}l_k=v_k^{\zeta}\beta_ke^{-\zeta_k};\ee
              \item\textbf{maximal household's consumption:} 
        \be\label{max-c}
         c^0_k=
        \frac{\gamma_ke^{-\xi_k}}{\psi(\zeta)};
            \ee
 \item\textbf{interest rates}: 
\be\label{rkopt}\theta_kr_k=e^{\zeta_k}-1; \ee



%

%

    \item \textbf{prices over wage:}
        \be\label{prices} \frac{p_k}{w}=\frac{e^{\xi_k}}{\E[e^{\rho_k}|\varphi(\eta)]}\,;\ee
\item \textbf{wage:}\be\label{wage}w=\prod_{k\in\mc V}p_k^{\gamma_k} \prod_{k\in\mc V}\E[e^{\rho_k}|\varphi(\eta)]^{\gamma_k} e^{-\sum_{k}v_k^0\zeta_k}\,.\ee
\item \textbf{prices normalized by consumption bundle ones:}
\be\label{price-consbundle}\frac{p_k}{\prod_jp_j^{\gamma_j}}=
\prod_{j\ne k}\E[e^{\rho_j}|\varphi(\eta)]^{\gamma_j} e^{\sum_{j}(L_{kj}-v_j^0)\zeta_j}\,.\ee
\end{itemize}
\end{theorem}
A few comments are in order. First, recall that all quantities in equations \eqref{max-y}--\eqref{wage} are measurable with respect to the public signal $\varphi(\eta)$, i.e., they are computable by all actors in the economy based on the information available at time $0$. 

As for quantities, Equation \eqref{max-y} states that the maximal production output of industry $k$ equals the product of its normalized discounted centrality $v^\zeta_k$ times the discount factor of sector $k$ computed in correspondence with the total cost of debt $\xi_k$. So, the maximal production is decreasing in the total cost of debt, and is equal to the Bonacich centrality index $v^0_k$ if the cost of debt is null for every sector, for instance,  when no sector is levered. 
Equation \eqref{max-z} states that the maximal quantity of product $j$ ordered by sector $k$ at time $0$ is  the product of its own normalized discounted centrality $v^\zeta_k$  times the importance $A_{jk}$ of product $j$ in industry $k$'s technology, times both the discount factor in correspondence with sector $k$'s primitive cost of debt and sector $j$'s total cost of debt. 
So, consistent with the overall production, the use of intermediate goods is also decreasing in the overall cost of debt, and it collapses into the product of the Bonacich centrality and the importance coefficient or labor share  when the cost of debt is null for every sector.   Equation \eqref{labor} states that the amount of labor employed by industry $k$ is equal to the product of the normalized discounted centrality $v^\zeta_k$ times the labor share $\beta_k$ times the discount factor for the primitive cost of debt. 
 Equation \eqref{max-c} states that the maximal consumption of product $k$ by the representative household equals the product of the preference weight $\gamma_k$ times  the discount factor for the total cost of debt of the sector $k$, divided by the normalization factor $\psi(\zeta)$. It is decreasing in the total cost of debt, as output and intermediate goods are, and reduces to the preference weight when  all sectors have zero cost of debt. All real quantities in equilibrium are then deeply affected by leverage through its price, as represented by the cost of debt both of the sector and its suppliers.

Regarding prices, including the interest rate, Equation \eqref{rkopt} recalls that the compound factor which includes the primitive cost of debt can be written as a discrete-time compound factor  accounting for the product of the leverage times the interest rate on sector $k$. This justifies our name, namely, cost of debt, for $\zeta$.  Equation \eqref{prices} states that the unit cost of product $k$ normalized by the wage at equilibrium equals the ratio between the compound factor in correspondence to the  total cost of debt on sector $k$ divided by the conditional expected total shock $\E[e^{\rho_k}|\varphi(\eta)]$. 

In contrast to quantities, unit prices of goods are increasing in the total cost of debt to the corresponding sector, for any given expectation of the total real shocks to the latter. Therefore, prices are at a minimum when the cost of debt for each sector is null. This points to a very different role of leverage in determining quantities and prices in the Walrasian equilibrium. Quantities decrease and prices increase in the overall cost of debt.   

Finally,  Equation \eqref{wage} connects the wage to the expected total shocks, discounted using the cost of debt of the different sectors weighted by their (non-discounted) Bonacich centrality. This expression has a maximum with respect to the cost of debt when the latter is null for every sector, independently of the network structure embedded in the Bonacich network centrality. Last but not least, \eqref{price-consbundle} examines  the price of each good $k$ relative to the one of the consumption bundle, and states that it is increasing  in the cost of debt of all sectors which are suppliers of $k$ - the ones we defined above as having $j \neq k$, $L_{kj}>0$ - and whose Leontief coefficient with respect to $k$ is greater than their own  Bonacich centrality. Also, these view of prices confirm the role of leverage in their relative level, 

Theorem \ref{theo:Walrasian} implies the following result which  characterizes  the actual or realized quantities.

\begin{corollary}\label{coro:Walrasian} 
Consider a constant returns-to-scale Cobb-Douglas economy $(A,\gamma)$ and leverage vector $\theta$. Let $\eta$ be a primitive log-productivity shock vector and  let $\varphi(\eta)$ be a public signal. Then, at the $\varphi$-rigid Walrasian equilibrium, we have the following:  
\begin{itemize}   
        \item\textbf{actual productions:}
\be\label{actual-y}y^\eta_k=v_k^{\zeta}e^{\rho_k-\xi_k}\,;\ee 
    \item\textbf{actual intermediate quantities:}
    \be\label{actual-z}z^\eta_{jk}=v_k^{\zeta} A_{jk}e^{\rho_j-\zeta_k-\xi_j}\,;\ee
              \item\textbf{actual household's consumption:}
\be\label{actual-c}c^\eta_k=
        \frac{\gamma_ke^{\rho_k-\xi_k}}{\psi(\zeta)}\,;   \ee
          \item\textbf{actual profits:}
             \be\label{actual-profit}  \pi_k(\eta)=\displaystyle \left(\tau_k -\eps_k \right) wv^{\zeta}_k \,;\ee
         \end{itemize}
        \begin{itemize}
    \item \textbf{actual budget:}
    \be\label{actual-budget}\mc E(\eta)=\frac{w}{\psi(\zeta)}\sum_{k\in\mc V}\gamma_k\tau_k\,;\ee
              \item\textbf{actual welfare:}
\be\label{actual-utility}            U(c^{\eta})=e^{\sum_{k\in\mc V}v^0_k\eta_k}\frac{e^{-\sum_{k\in\mc V}v^0_k\zeta_k}}{\psi(\zeta)}\,.\ee
            \end{itemize}
            \end{corollary}

Observe that Equations \eqref{actual-budget} for the actual budget and Equation \eqref{Etauk=1} for the conditional expected information-normalized total shocks, along with the constant returns-to-scale assumption \eqref{normalization-2} for the final consumer preferences,  imply that 
$$\E[\mc E(\eta)|\varphi(\eta)]=\frac{w}{\psi(\zeta)}\sum_{k\in\mc V}\gamma_k\E[\tau_k|\varphi(\eta)]=\frac{w}{\psi(\zeta)}\sum_{k\in\mc V}\gamma_k=\frac{w}{\psi(\zeta)}\,,$$
so that we can interpret the normalization factor $\psi(\zeta)$ introduced in Equation \eqref{psi} as the ratio between the wage and the conditional expected budget given the public signal.

           \subsection{The full information benchmark}\label{sec:full}
In the special case when there is full information available, i.e., when the public signal observable by all the actors in the economy $\varphi(\eta)=\eta$ coincides with the log-productivity shock vector itself, we get that the primitive cost of debt  is null for every sector, i.e., \be\label{zetak=0}\zeta=0\,,\ee so that also $\xi=0$, i.e., the total cost of debt is null for every sector. As a consequence, in this special case, Equations \eqref{max-y}--\eqref{rkopt} for the maximal production, maximal quantities of intermediate goods ordered at time $0$, employed labor, and loan rates reduce to 
        \be\label{0distortion-1}y^0_k  =v^0_k\,,\qquad
         z^0_{jk}=v^0_k A_{jk}\,,\qquad
         l_k=v_k^0\beta_k\,,\qquad
         c^0_k=\gamma_k\,,\qquad
        \theta_kr_k=0\,,\ee
             while Equations \eqref{actual-y}--\eqref{actual-c} for the actual productions and consumptions reduce to 
          \be\label{0distortion-3}   y^\eta_k  =v^0_ke^{\rho_k}\,,\qquad
         z^\eta_{jk}=v^0_k A_{jk}e^{\rho_j}\,,\qquad
           c^{\eta}_k=\gamma_ke^{\rho_k}\,. \ee   
Moreover, we have that $\E[e^{\rho_k}|\varphi(\eta)]=e^{\rho_k}$ and $\tau_k=\eps_k=1$ for every sector $k\in\mc V$, so that Equation \eqref{prices}  reduces to  
\be\label{0distortion-2bis}            p_k=we^{-\rho_k}\,,\ee
i.e., prices are state-contingent. Also the wage is state-contingent, since Equation \eqref{wage} becomes \be w=\prod_{k\in\mc V}p_k^{\gamma_k}\prod_{k\in\mc V} e^{\rho_k \gamma_k}=\prod_{k\in\mc V}p_k^{\gamma_k}e^{\sum_kv^0_k \eta_k}\,.\ee
Equation \eqref{actual-profit} reduces to  
\be\label{0distortion-4bis}            \pi_k(\eta)=0\,,\ee
i.e., profits are identically $0$. 
Equation \eqref{actual-budget} reduces to 
\be\label{E=w}\mc E(\eta)=w\,,\ee
i.e., the actual budget coincides with the total wage, and Equation \eqref{actual-utility} reduces to 
\be\label{U=ev0}U(c^\eta)=e^{\sum_{k\in\mc V}v^0_k\eta_k}\,.\ee
It follows that the Domar weight $\lambda_k$ of each sector $k\in\mc V$, defined as the ratio  of the sales ---or actual revenues--- over the total consumption, satisfies 
\be\label{Domar-full-info}\lambda_k:=\frac{p_ky^{\eta}_k}{\mc E(\eta)}=v^0_k\,,\ee 
i.e., it coincides with the Bonacich centrality of sector $k$. In particular, Domar weights are deterministic in this special case. Now, observe that, from Equation \eqref{actual-utility} for the actual welfare we get
\be\label{pd-welfare}\frac{\partial}{\partial \eta_k}\log U(c^{\eta})=v^0_k\,,\ee
i.e., the partial derivative of the log-welfare with respect to the log-primitive shock in sector $k$ coincides with the sector's Bonacich centrality. 
It follows from Equations \eqref{Domar-full-info} and \eqref{pd-welfare} that, under full information, and as expected, Hulten's theorem applies, as the Domar weight of sector $k$ measures  the first order effect of a shock to that  sector on aggregate welfare: 
\be
\frac{\partial}{\partial \eta_k}\log U(c^\eta)= v^0_k= \lambda_k\,.
\ee

 Hence, in the full information case, we recover the standard results already present in the literature, as reported for instance in Carvalho and Tazbah-Salehi \cite[Sec.\ 2]{carvalho2019production}. The novelty is that it applies with leverage.
 
\section{Information frictions effects: Hulten's theorem's failure}\label{sec:effects}\label{sec: hulten} 
We now consider the case when information is not full, i.e., the log-productivity shock vector $\eta$ is not measurable with respect to the public signal $\varphi(\eta)$.


 Using Equation \eqref{prices} in Theorem \ref{theo:Walrasian} for the equilibrium prices, and Equations \eqref{actual-y} and \eqref{actual-budget} in Corollary \ref{coro:Walrasian} for the actual production and consumer budget or GDP in equilibrium, it is straightforward to express the Domar weight of each sector $k\in\mc V$ as
\be \label{DW} \lambda_k=\frac{p_ky^{\eta}_k}{\mc E(\eta)}= \frac{\psi(\zeta)v^{\zeta}_k\tau_k   }{\sum\limits_{j\in\mc V} \tau_j \gamma_j}\,.\ee
Since Equation \eqref{pd-welfare} for the partial derivatives of the log-consumer utility with respect to the production shocks continues to hold true in the general case of partial information, in general Hulten's theorem does not hold. In fact, the first-order effect of a shock is still measured by the Bonacich centrality,  but  the latter does not coincide with the Domar weight anymore:  
\begin{equation}\label{noH}
 \frac{\partial }{\partial \eta_k}\log U(c^\eta) =  v^0_k \neq \psi(\zeta)\frac{v^{\zeta}_k\tau_k   }{\sum\limits_{j\in\mc V} \tau_j \gamma_j}= \lambda_k   \,.
\end{equation}
This means that the drawback of having Domar weights that measure the first-order effect of a shock on aggregate output is over: even sectors that have a small Domar weight on total GDP can significantly affect welfare. They do so the higher their centrality, as intuition would suggest, but their centrality does not coincide with their Domar weight. 

We have the following result, proved in Appendix \ref{sec:proof-prop-no-hulten}.  

\begin{theorem}\label{prop:no-hulten}
Consider a constant returns-to-scale Cobb-Douglas economy $(A,\gamma)$ and leverage vector $\theta$. Let $\eta$ be a primitive log-productivity shock vector and  let $\varphi(\eta)$ be a public signal. Then, Hulten's theorem holds true with probability $1$ at the $\varphi$-rigid equilibrium if and only if the cost of debt is null in every sector, i.e., Equation \eqref{zetak=0} holds true, and there exist a scalar random variable $Y$ and a measurable function $g:\mc S\to\R^{\mc V}$ such that 
\be\label{eta=Y}\eta=Y\beta+g(\varphi(\eta))\,.\ee
     If this is not the case, then \begin{equation} \label{dlogU>lambda}\frac{\partial }{\partial \eta_k}\log U(c^\eta)>\lambda_k\,,
\end{equation} for some sector $k\in\mc V$  with positive probability. 
\end{theorem}
Theorem \ref{prop:no-hulten} states that Hulten's theorem fails to hold true whenever $\zeta\ne0$, i.e., if the cost of debt is positive in at least one sector (something that occurs whenever information is not full and leverage is positive in at least one sector) or if $\zeta=0$ and the random vector $\eta$ of the log-productivity shocks cannot be decomposed as the sum of a random multiple of the vector $\beta$ plus a random vector $g(\varphi(\eta))$ that is measurable with respect to the public signal $\varphi(\eta)$.



This result is extremely important in understanding supply chain disruptions, such as the ones we observed during the COVID-19 pandemics or as a consequence of wars and civil unrest, all over the world. We witnessed that the effect of a shock on GDP is not related to the affected sector's size, but rather to its connections.\footnote{See, e.g., the Toyota case illustrated in Elliot \& Jackson \cite{elliott2024supply}} The model rationalizes this.

In our model, the relationship between a sector's Bonacich centrality $v^0_k$ and its Domar weight $\lambda_k$ ---namely, whether the former exceeds the latter, and therefore whether the effect of a shock is larger than Hulten's theorem would command--- depends  on leverage, for any given consumer proximity, as measured by the preference vector $\gamma$ and realization of the log-productivity shock vector $\eta$. 
Substituting for the Bonacich measures in the Domar weight we get that 
\eqref{dlogU>lambda} holds true, 
if and only if 
\begin{equation} \label{no hult}
 \frac{\tau_k}{ \sum\limits_{j\in\mc V} \tau_j \gamma_j}<\frac{v^0_k}{\psi(\zeta)v_k^{\zeta}}\,,
\end{equation}
i.e., a sector $k$'s first-order effect on welfare exceeds its Domar weight whenever the ratio between sector $k$'s normalized total shock $\tau_k$ and the average $\sum_j\gamma_j\tau_j$ of all sectors' total shocks, is smaller than  the ratio of the averages of its Leontief coefficients without and with leverage. All the averages are weighted by the preference coefficients. 

If inequality \eqref{no hult} holds true, it will remain valid if consumer proximities and shock realizations are kept constant, but leverage is increased. This is because the only changing term, the denominator of the right hand side, is monotone non-increasing in $\zeta$ thanks to Lemma \ref{lemma:monotonicity-zeta} (iii) and $\zeta$ is non-decreasing in  leverage, as shown in Proposition \ref{prop-zeta}.

\section{Leverage effects with information frictions}
This Section examines the effects of leverage: the most important result is that the combination of leverage and information frictions causes the departure from the first best. We saw in Section \ref{sec:full} above that, as long as there is leverage without information frictions, there is no loss in efficiency. To get to the complete result, in Section \ref{sec: unlevered} we examine an economy with information frictions, but without leverage, and show that it does not cause departures from the first best. In Section \ref{sec: lev cons} we examine the case with information frictions and leverage and comment on the loss of efficiency, as well as its impact through the network.
\subsection{The unlevered economy}\label{sec: unlevered}
 When all sectors are unlevered, $\theta=0$,  but information is not perfect, then the vector of the costs of debt $\zeta$ is null, as proved in Section \ref{sec:eq}, and the total cost of debt is null too, $\xi =0$. The previous Section proved that Hulten's theorem fails to hold, unless the shocks - letting aside a function of the information signal -  are linear in the vector of labor usage (with a random coefficient). Letting this case aside, even an unlevered but rigid economy does not incur into the restrictions of Hulten's theorem. Apart from that theorem, the equilibrium is characterized by 
\begin{itemize}    
\item\textbf{maximal productions:}
\be y^0_k 
           =v_k^{0}\,;   \ee   
    \item\textbf{maximal intermediate quantities:}
     \be\ z^0_{jk}
         =v_k^{0} A_{jk};\ee
    \item\textbf{employed labor:}
\be l_k=v_k^{0}\beta_k;\ee
              \item\textbf{maximal household's consumption:} 
        \be
         c^0_k=
        \gamma_k;
            \ee

    \item \textbf{prices over wage:}
        \be \frac{p_k}{w}=\frac{1}{\E[e^{\rho_k}|\varphi(\eta)]}\,;\ee
\item \textbf{ wage:} \be w=\prod_{k\in\mc V}p_k^{\gamma_k} \prod_{k\in\mc V}\E[e^{\rho_k}|\varphi(\eta)]^{\gamma_k} \,.\ee
\item \textbf{prices normalized by consumption bundle ones:}
\be \frac{p_k}{\prod_jp_j^{\gamma_j}}=
\prod_{j\ne k}\E[e^{\rho_j}|\varphi(\eta)]^{\gamma_j} \,.\ee
\end{itemize}
while the actual quantities and prices  are
\begin{itemize}   
        \item\textbf{actual productions:}
\be y^\eta_k=v_k^{0}e^{\rho_k}\,;\ee 
    \item\textbf{actual intermediate quantities:}
    \be z^\eta_{jk}=v_k^{0} A_{jk}e^{\rho_j}\,;\ee
              \item\textbf{actual household's consumption:}
\be c^\eta_k=
        \gamma_ke^{\rho_k}\,;   \ee
          \item\textbf{actual profits:}
             \be  \pi_k(\eta)=\displaystyle \left(\tau_k -\eps_k \right) wv^{0}_k \,;\ee
         \end{itemize}
        \begin{itemize}
    \item \textbf{actual budget:}
    \be \mc E(\eta)= w\sum_{k\in\mc V}\gamma_k\tau_k\,;\ee
              \item\textbf{actual welfare:}
\be           U(c^{\eta})=e^{\sum_{k\in\mc V}v^0_k\eta_k}\,.\ee
            \end{itemize}
So, the equilibrium with rigidity,  and without leverage and full information remains different from the full information one. Prices are non-random, while they are state contingent in both the full info and the rigid, levered case. Welfare  remains the same as in the full information case, thanks to prices adjustments. Sectors however may have negative as well as positive profits, something which full information prevents.

\subsection{The levered economy}\label{sec: lev cons}
We now investigate an economy where information is not full and leverage $\theta_k>0$ is positive for some sector $k\in\mc V$. These two features cause the  cost of debt of that sector to be positive, i.e., $\zeta_k\ne0$: We study not only the effect on welfare, but also on intermediate goods and labor usage, as well as consumption and prices.

First, the combination of leverage and rigidity causes the departure from the first best.  The following result is proved in Appendix \ref{sec:proof-theo-welfare}.   
\begin{theorem}\label{theo:welfare}
Consider a constant returns-to-scale Cobb-Douglas economy $(A,\gamma)$ and leverage vector $\theta$. Let $\eta$ be a primitive log-productivity shock vector and  let $\varphi(\eta)$ be a public signal. Then, the welfare  at the $\varphi$-rigid  equilibrium satisfies  
$$U(c^\eta)\le e^{\sum_{k\in\mc V}v^0_k\eta_k}\,,$$
with equality if and only if $\zeta=0$. 
\end{theorem}
Theorem \ref{theo:welfare} states that the maximal welfare is achieved if and only if the cost of debt is zero for every sector $k$. A positive cost of debt in some sectors, due to the combination of positive leverage in these sectors and partial information, can only decrease the welfare at equilibrium. We argue below that the inefficiency is due to an inefficient use of resources. 
Indeed, theorem \ref{theo:Walrasian} proved that, in the presence of information frictions,  leverage ---through the total cost of debt---  decreases maximal quantities of intermediate goods and labour but increases maximal prices. However, to assess whether  and how the inefficient use of resources causes the loss in aggregate utility, we must  specify how the effects on maximal quantities and prices depend on the  leverage of each single sector  and are distributed along the supply chain.

Proposition \ref{prop:lev-effect-qp} below - which  is proved in Appendix \ref{sec:proof-prop-lev-effect-qp} - states that leverage increases maximal consumption and labor upstream. It has inflationary effects and generates an increase of the total cost of debt on the levered sector, which propagates downstream.  

\begin{proposition}\label{prop:lev-effect-qp}
   Consider a constant returns-to-scale Cobb-Douglas economy $(A,\gamma)$ and leverage vector $\theta$. Let $\eta$ be a primitive log-productivity shock vector and  let $\varphi(\eta)$ be a public signal. If the leverage $\theta_o$ of a sector $o$ increases, then, at the $\varphi$-rigid  equilibrium:  
    \begin{enumerate}
    \item[(i)] the total cost of debt of sector $o$ and of all its customers does not decrease, while the total cost of debt of every other sector is not affected;
    \item[(ii)] the labor employed by every industry other than sector $o$ and its suppliers does not decrease and the employed labor  by at least one sector among $o$ and its suppliers does not increase;
        \item[(iii)] the  maximal consumption of every industry other than sector $o$ and its costumers does not decrease; 
        \item[(iv)] the price over wage ratio of sector $o$ and all its customers does not decrease 
        and the unit price $p_k$ of the $k$-th good normalized by the consumption bundle one $\prod_jp_j^{\gamma_j}$ does not decrease for every sector $k$ such that $L_{ko}\ge v^0_o$ and does not increase for every sector $k$ such that $L_{ko}\le v^0_o$.
\end{enumerate}
\end{proposition}
In summary, even before shocks arrive, leverage, together with incomplete information, has distortionary effects on the use of the resources, similarly to the distortions in \cite{bigio2020distortions}. This explains the impact on welfare. Proposition \ref{prop:lev-effect-qp} qualifies the position (suppliers, customers) that are affected. It also provides the condition ---on the Leontief coefficient of any sector with respect to the levered one, $L_{ko}$--- for inflation to arise. 


\section{Defaults\label{sec:def}}
In this economy, sectors can have a negative profit, namely they can experience a default event, and therefore they can affect financiers, potentially causing losses to them. In this section we aim to investigate defaults more in detail, showing their dependence on various factors such as the production network structure encoded in the Leontief matrix, the shock distribution and its specific realization, the information signal, and the leverage.
We will study the phenomenon of amplification, namely the mechanism by which a production shock in a certain part of a production network can trigger default events in close by, but possibly also in far away, nodes. 
We will also investigate the magnitude of the losses once default occurs.

We notice from formula \eqref{actual-profit} that  the profit is the product of two expressions,  the difference $(\tau_k -\eps_k)$ and the product of the wage times the normalized discounted network centrality $v^{\zeta}_k$. The first one   depends on the shock realization $\eta$ and may take positive or negative values. It does neither depend on leverage, nor on  consumer preferences. The second one has strictly positive sign and depends on the shock only through the vector $\zeta$ of the primitive costs of debt. It is measurable with respect to $\varphi(\eta)$ and depends on the leverage $\theta$.
The default event for sector $k\in\mc V$, defined as $\{\pi_k(\eta)<0\}$, only depends on the first expression. Indeed, we have that $\pi_k(\eta)<0$ if and only if 
\be\label{default}\tau_k<\eps_k\,.\ee
Considering the way $\eps_k$ is defined in \eqref{epsk}, we understand  that default for sector $k$ happens when the total shocks to its direct suppliers, appropriately weighted,  are greater than its owns. The comparison depends in both cases on total, not primitive shocks.

Once the default event occurs,  \eqref{actual-profit} implies that the magnitude of the losses depends  also on the product of the wage and the discounted centrality index. We can interpret this product, $w  \nu^\zeta_k$, when it is greater than one, as a magnifier of losses (and profits, whenever the first difference is positive) and therefore we interpret it as determining the \emph{amplification} of shocks into losses. 




\subsection{Defaults from a single shock and cascades}

To explore further the occurrence of default, let us focus on the case when the shock originates in a single node, that, from now on, we indicate with the symbol $o$.  This amounts to assume that $\eta_k=0$
deterministically, for every node $k\neq o$.  The assumption allows us to isolate the propagation of a given default without the confounding effects of other failures. We say that a \emph{cascade} occurs when some sector different from $o$ also defaults. this section examines whether the cascade happens, and how much it extends, because sectors other than the initial one default (Proposition \ref{prop:convexity}), whether the sector itself does (Proposition \ref{prop:convexity2}), and when amplification occurs. To prepare for the numerical examples below, we conclude by  specifying the conditions for cascades and default of the shocked sector when the shock  distribution is exponential.

To present our analytical  results, it is convenient to introduce the following $\varphi(\eta)$-measurable objects
\be\label{tilted-mom}m_k(t):=\frac{\E[\eta_o^ke^{t\eta_o}|\varphi(\eta)]}{\E[e^{t\eta_o}|\varphi(\eta)]},\quad k=1,2, \qquad \sigma^2(t):=m_2(t)-(m_1(t))^2.
\ee
They can be interpreted as the first two $\varphi(\eta)$-conditioned momenta of a random variable $\tilde\eta_o$ (known as the exponential tilting of $\eta_o$) absolutely continuous with respect to $\eta_o$, with Radon-Nikodym derivative   
$$x\mapsto \frac{e^{t x}}{\E[e^{t \eta_o}|\varphi(\eta)]}.$$
In particular, for $t=0$, expressions in \eqref{tilted-mom} coincide with, respectively, the $\varphi(\eta)$-conditioned momenta  of $\eta_o$. 

We also define, for every sector $k$, the quantity
$$L^-_k:=\max\{L_{jo}\,|\, j\in\mc V,\, A_{jk}>0\},$$
which is the maximum entry of column $o$ of the Leontief matrix, and therefore the maximum  coefficient which links any sector $j$ to $o$, among those sectors $j$ which are suppliers of $k$. The following result, which is proved in Appendix \ref{sec:proof-prop-convexity},  holds true.

\begin{proposition}\label{prop:convexity} Consider a constant returns-to-scale Cobb-Douglas economy $(A,\gamma)$ and leverage $\theta$. 
Let $\eta$ be a primitive log-productivity shock vector concentrated only on node $o$ (i.e., $\eta_k=0$
deterministically, for every node $k\neq o$). 
Then, we have that, for $k\neq o$:
\begin{itemize}
\item Every sector $k$ such that $L_{ko}=0$ is never in default, i.e., $\pi_k(\eta)=0$ deterministically.
\item If the shock realization satisfies 
\be\label{inside}\eta_o\in [m_1(t)-\sigma(t), m_1(t)+\sigma(t)], \quad \forall 0\leq t\leq \bar t,\ee then every sector $k\neq o$ for which $L_k^-\leq \bar t$ is not in default.
\item If the shock realization satisfies \be\label{notinside}\eta_o\not\in [m_1(t)-\sigma(t), m_1(t)+\sigma(t)], \quad\forall t\geq 0,\ee then every sector $k$  such that $L_{ko}\neq 0$ is in default.
\end{itemize}
\end{proposition}
As far as sectors $k\ne o$ are concerned, the first result says that, as intuition would suggest, every sector $k$ for which there is no walk in the network from the shock source $o$ to the sector itself will never experience a default.

The second result says that all sectors, whose suppliers have Loentief coefficients with respect to the shock source sufficiently small, will not default. The corresponding condition \eqref{inside} guarantees that the default event will have a limited impact on the production network. Indeed, condition \eqref{inside} guarantees that the shock realization is sufficiently close to its expected value (for the basic shock distributions as well for a family of the exponential tilting versions). A continuity argument implies that if we just have 
\be\label{inside0}\eta_o\in (m_1(0)-\sigma(0), m_1(0)+\sigma(0)),\ee then \eqref{inside} holds for some $\bar t>0$ and thus all sectors with sufficiently small $L_k^-$ will not default.  
The second result then says that shocks that are not too disruptive possibly generate a partial cascade of defaults.

The third result, condition \eqref{notinside},  is a sufficient condition for the size of the shock to be such that  default takes place throughout the network, whatever the production network is. Here the default cascade affects the whole network.

Proposition \ref{prop:convexity} does not give any information regarding the default event of the node $o$ primarily hit by the shock. Indeed, the analysis for such a node is somewhat peculiar and leads to different results that are gathered in the proposition below.    

\begin{proposition}\label{prop:convexity2} Consider a constant returns-to-scale Cobb-Douglas economy $(A,\gamma)$ and leverage $\theta$. 
Let $\eta$ be a primitive log-productivity shock vector concentrated only on node $o$ (i.e., $\eta_k=0$
deterministically, for every node $k\neq o$). 
If $L_{oo}=1$ (i.e., $o$ does not belong to any cycle), then node $o$ is in default if and only if the shock realization satisfies
\be\label{default-o} e^{\eta_o}<\E[e^{\eta_o}\,|\,\varphi(\eta)].\ee
In the general case, we have that:
\begin{itemize}
\item If the shock realization satisfies 
\be\label{inside2}\eta_o\in [m_1(t), m_1(t)+\sigma(t)], \quad \forall 0\leq t\leq L_o^-+1,\ee then sector $o$ is not in default.
\item If the shock realization satisfies \be\label{notinside2}\eta_o<m_1(t)-\sigma(t), \quad \forall t\geq 0,\ee then sector $o$ is in default.
\end{itemize}
\end{proposition}

To provide an intuition for these results, recall that primitive  log shocks in $\eta$ take non-positive values and the greater is the shock $\eta_o$ in absolute value, the smaller is its exponential. So, (\ref{default-o}) guarantees that ---in the absence of cycles--- default occurs if and only if the primitive log shock is smaller than a given threshold. When cycles exacerbate the consequences of the shock, and back fire, then the sector originally hit is not in default if the actual shock stays at most one standard-deviation above the mean, and it does default when it is at least one standard-deviation  below the expectation, which intuitively means that occurrences are very bad. The last two condition must hold  for the mean and standard deviation of both the shock $\eta_o$ and its exponential transforms.

Let us examine amplification. To answer the question of whether amplification occurs downstream, once a default occurs,  let us recall parts (iv) and (vi) of Lemma \ref{lemma:monotonicity-zeta}.  Those properties indeed imply that, when  a sector $k$ is part of the cascade, amplification from a realized shock  $\tau_k - \eps_k$ to the sector losses, represented by $w v_k^\zeta >1$, if it exists,  is  
\begin{itemize}
\item  non-decreasing in the cost of debt of sector $o$, $\zeta_o$, for the sectors which are not supplier of $o$, in the sense established by part (iv) of that Lemma. 
\item  increasing in the cost of debt of sector $o$, $\zeta_o$, for the sectors that are not suppliers of $o$, in the sense established by part (vi) of that Lemma.
\end{itemize}
Loosely said, if amplification occurs, it is greater for sectors downstream to $o$.

To conclude this discussion and to prepare for our examples in the next Section, let us now focus on a specific case: (i) no information (i.e., $\varphi(\eta)=C$) and (ii) $-\eta_o$ is an exponential random variable with parameter $\lambda>0$. 
In this case, we can explicitly compute
$m_1(t)=-\sigma(t)=1/(t+\lambda)$ so that
$$[m_1(t)-\sigma(t), m_1(t)+\sigma(t)]=\left[-\frac{2}{t+\lambda}, 0\right].$$
We can reformulate the results of Propositions \ref{prop:convexity} and \ref{prop:convexity2} as follows. 
\begin{corollary}\label{prop:convexity-exp} Consider a constant returns-to-scale Cobb-Douglas economy $(A,\gamma)$ and leverage $\theta$. Assume that there is no information.  
Let $\eta$ be a primitive log-productivity shock vector concentrated only on node $o$ (i.e., $\eta_k=0$
deterministically, for every node $k\neq o$) and let $-\eta_o$ be an exponential random variable with parameter $\lambda>0$. 
Then, we have that:
\begin{enumerate}
\item[(i)] every sector $k$ such that $L_{ko}=0$ is never in default, i.e., $\pi_k(\eta)=0$ deterministically;
\item[(ii)] if the shock realization satisfies $\eta_o>-2/\lambda$, then every sector $k\neq o$ for which $L_k^-\leq -2/\eta_o -\lambda$ is not in default;
\item[(iii)] if the shock realization satisfies $\eta_o\leq -2/\lambda$, then every sector $k$ such that $L_{ko}\neq 0$ is in default.
\end{enumerate}
\end{corollary}

\vspace{10pt}

\section{Examples and numerical simulations\label{sec:ex}}

In this section, we analyze in detail two basic topologies and compute all the relevant quantities of the model, with a particular focus on default.


\subsection{Directed Line}\label{sec:ex-line}
We first study the case of a production network with $n=4$ sectors, forming a directed line graph with a root node $1$ and subsequent nodes $2,3,4$ (see Figure \ref{fig:line}).
\begin{figure}
\begin{center}
\begin{tikzpicture}[
    >=stealth,
    every node/.style={circle, draw, minimum size=10mm, inner sep=0pt}
]

    \node[draw=green!60!black, fill=green!20] (n1) at (0,0) {$1$};
    \node[draw=red, fill=red!20]              (no) at (2,0) {$2$};
    \node[draw=gray, fill=gray!20]            (n3) at (4,0) {$3$};
    \node[draw=gray, fill=gray!20]            (n4) at (6,0) {$4$};

    \draw[->] (n1) -- (no);
    \draw[->] (no) -- (n3);
    \draw[->] (n3) -- (n4);

\end{tikzpicture}
\end{center}
\caption{\label{fig:line}}
\end{figure}
Specifically, we assume that  $A_{12}=A_{23}=A_{34}=\alpha$ in $(0,1)$, and $A_{ij}=0$ for all $(i,j)\notin\{(1,2),(2,3),(3,4)\}$, and that consumer preferences are uniform, i.e., $\gamma_k=1/4$ for every sector $k$ in $\mc V$. Since the matrix $A$ is nilpotent with $A^4=0$, the economy's Leontief inverse is given by
\begin{equation}
L =(I-A')^{-1}=I+A'+(A')^2+(A')^3=
\begin{bmatrix}
1 & 0 & 0 & 0\\
\alpha & 1 & 0 & 0\\
\alpha^2 & \alpha & 1 & 0\\
\alpha^3 & \alpha^2 & \alpha & 1 
\end{bmatrix},
\end{equation}
and the Bonacich centralities are 
$$v^0_1=\frac{1+\alpha+\alpha^2+\alpha^3}{4},\qquad v^0_2=\frac{1+\alpha+\alpha^2}{4}\,,\qquad v^0_3=\frac{1+\alpha}4\,,\qquad v^0_k=\frac14\,.$$

We focus on the special case where the log-productivity shock hits only the single node $o=2$, so that $\eta_1=\eta_3=\eta_4=0$ are deterministic, and that $-\eta_2$ is exponentially distributed with parameter $\lambda>0$. 
 The corresponding total log-productivity shock vector is then $$\rho =L\eta=\eta_2(0,1,\alpha,\alpha^2)\,. $$
 Notice that the total log-productivity shock in node $1$ is $\rho_1=0$, consistently with the fact that production shocks propagate only downstream, while the magnitude of the total log-productivity shock achieves its maximum value in node $2$, where the primitive  log-productivity production shock is concentrated and decays moving further down the chain. 

 We now assume that the public signal is constant, so that the actors in the economy have no information about the production shock vector realization when they make their decisions at time $0$, but they only know its prior distribution. Observe that in this case, Corollary \ref{prop:convexity-exp} implies that: 
\begin{itemize}
\item industry $1$ never defaults, as its profit $\pi_1(\eta)$ is deterministically $0$; 
\item if $\eta_2<-2/\lambda$, then industries $2$, $3$, and $4$ all default;
\item if $-2/\lambda\leq \eta_2<\log{\lambda}/(1+\lambda)$, then the set of sectors that default is either $\{2\}$, or $\{2,3\}$, or $\{2,3,4\}$.
\end{itemize}
 
 Moreover, as shown in Appendix \ref{sec:proof-sec7}, it can be proved that
 \begin{itemize}
\item if $\eta_2\ge \log{\lambda}/(1+\lambda)$, then no sector defaults.  \end{itemize}





We now focus on the special case when $\alpha=0.6$ and $\lambda=0.25$. 
We study two different leverage regimes: (a) $\theta_k=0.5$ for every sector $k$; and (b) $\theta_k=1$ for every sector $k$.
Table \ref{tab_optimal_line} reports the maximal equilibrium quantities for each sector. It lists production $y^0$, employment $l$, consumption $c^0$ and relative prices $p/w$.

\begin{table}
\caption{Optimal Quantities in Equilibrium, Directed Line}
\label{tab_optimal_line}
\centering
{%
\begingroup\catcode`\_=12\relax
\begin{tabular}{|c|r|r|r|r|r|}
\hline
Leverage & Sector & y & l & c & p/w \\ \hline
\hline
0.5 & 1 & 0.54 & 0.54 & 0.33 & 1 \\ \hline
0.5 & 2 & 0.36 & 0.14 & 0.19 & 8.62 \\ \hline
0.5 & 3 & 0.34 & 0.19 & 0.22 & 5.16 \\ \hline
0.5 & 4 & 0.24 & 0.12 & 0.24 & 3.33 \\ \hline
1 & 1 & 0.55 & 0.55 & 0.36 & 1 \\ \hline
1 & 2 & 0.31 & 0.12 & 0.17 & 10.73 \\ \hline
1 & 3 & 0.31 & 0.2 & 0.2 & 6.14 \\ \hline
1 & 4 & 0.23 & 0.13 & 0.23 & 3.78 \\ \hline
\hline
\end{tabular}
 
\endgroup
}
\end{table}

As the leverage increases from $\theta_k=0.5$ to $\theta_k=1$ in every sector $k$, the economy exhibits a clear reallocation pattern. Sector $1$, which is not affected by the shock, experiences an increase in production, labor, consumption, and no effect on relative price. Consistently with Proposition \ref{prop:lev-effect-qp},  consumption (and production) fall in all other sectors, while labor falls in sector $2$, which is the one directly hit by the shock, but increases in the others. Prices rise, with the exception of sector $1$, highlighting the inflationary pressure generated by higher leverage. In general, increased leverage compresses real activity and consumption while  amplifying price differences across sectors.


Table~\ref{tab_financial_line} reports the marginal and total sector costs of debt  
$\zeta_k$, $\xi_k$. Both are zero for sector $1$, which cannot default. For the other sectors, the cost of debt $\zeta_{k}$ increases as $\theta$ increases,  as predicted by Proposition \ref{prop-zeta}. The same holds for the total cost of debt $\xi_k$, which coincides with the marginal value in sector $o=2$, and becomes relatively high in the other sectors despite their low primitives, as they may inherit the shock originating upstream.  The table also reports the (un-discounted) Bonacich centrality measures $\nu_k^{0}$ and the normalized discounted Bonacich centrality measures $\nu_k^{\zeta}$. The former are always smaller than the latter, with the partial exception of  sector 1, whose marginal and total cost of debt is 0 and which does not have the other, potentially defaulting sectors as suppliers. The discounted centrality of all sectors  increases, consistent with Lemma \ref{lemma:monotonicity-zeta}(vi), when leverage increases.

\begin{table}
\caption{Financial Variables, Directed Line}
\label{tab_financial_line}
\centering
{%
\begingroup\catcode`\_=12\relax
\begin{tabular}{|c|r|r|r|r|r|}
\hline
$\theta_k$ & $k$ & $\zeta_k$ & $\xi_k$ & $v^{0}_k$ & $v^{\zeta}_k$ \\ \hline
\hline
0.5 & 1 & 0 & 0 & 0.54 & 0.54 \\ \hline
0.5 & 2 & 0.55 & 0.55 & 0.49 & 0.61 \\ \hline
0.5 & 3 & 0.09 & 0.42 & 0.4 & 0.52 \\ \hline
0.5 & 4 & 0.06 & 0.31 & 0.25 & 0.33 \\ \hline
1 & 1 & 0 & 0 & 0.54 & 0.55 \\ \hline
1 & 2 & 0.77 & 0.77 & 0.49 & 0.66 \\ \hline
1 & 3 & 0.13 & 0.59 & 0.4 & 0.56 \\ \hline
1 & 4 & 0.08 & 0.44 & 0.25 & 0.36 \\ \hline
\hline
\end{tabular}
 
\endgroup
}
\end{table}

In order to explore the possible actual realizations, as opposed to the nominal quantities, we perform  $M= 100000$ Monte Carlo simulations of the shocks to sector o and compute all actual quantities in correspondence to them.

Figure \ref{fig_tau_epsilon_line} plots the overlapped histogram of the shocks and neighbor shocks $\eps_k$ and $\tau_k$ for $k=3,4$. Obviously, the total log-productivity shock in sector $1$ is always $\rho_1=0$, and consequently  $\tau_1=1$. Also, since sector $1$ does not have suppliers, $\eps_1=1$. The information-normalized total shock in the suppliers of sector $o=2$ is  $\eps_2=\beta_2+\tau_1A_{12}=1-\alpha+1-\alpha=1$ under all circumstances. This does not happen for the other sectors because they can indeed receive the shock from their suppliers.

\begin{figure}
\centering
\includegraphics[width=0.85\textwidth]{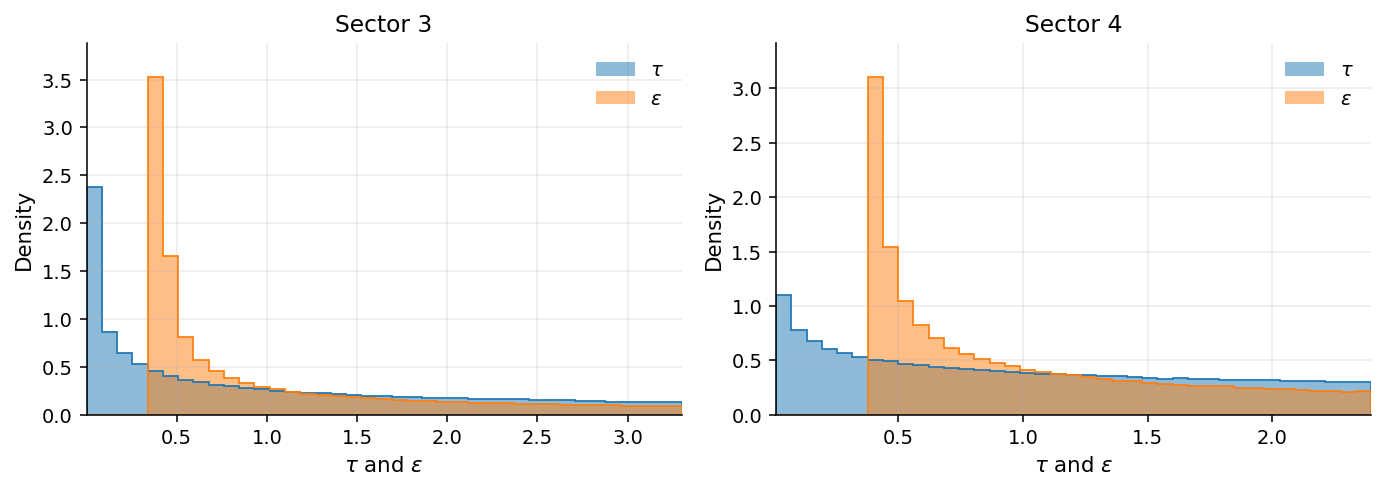}
\caption{Histogram of $\eps$ and $\tau$ by sector, $M=1000000$ draws, Directed Line.} \label{fig_tau_epsilon_line}
\end{figure}

Figure \ref{profit_hist} explores the economic effect of the  shocks to sectors $k=3,4$ by overlapping the simulated profit/loss distributions for $\theta=0.5$ (blue) and $\theta=1$ (orange). 

The higher leverage case ($\theta=\1$) exhibits a  wider range of profits/losses, indicating greater volatility in the  sector outcomes. 
This higher standard deviation reflects the amplified sensitivity of profits to shocks under increased financial exposure.

\begin{figure}
\centering
\includegraphics[width=0.85\textwidth]{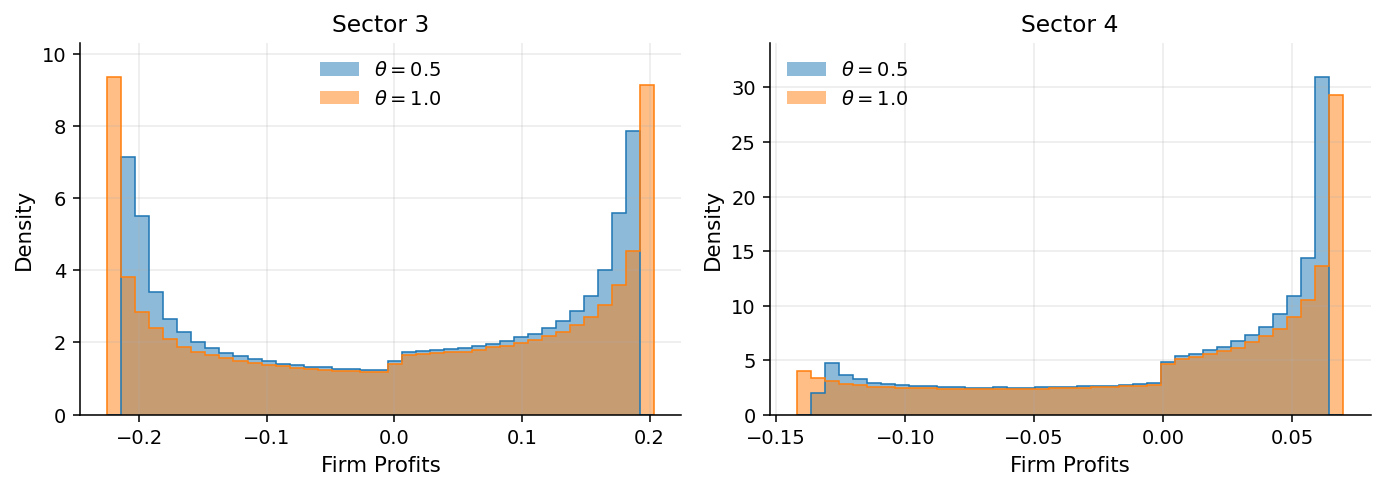}
\caption{Histogram of Sectors Profits, $M=1000000$ draws, Directed Line.} \label{profit_hist}
\end{figure}

To confirm this, Table \ref{tab_profits_line} compares the standard deviation (SD) of the sector profits  for the two different leverage values. The average profits for sectors are not reported, since they are zero across all sectors and all leverages, as due. However, the SD of profits shows a clear difference between the two leverage levels. The increased volatility and higher risk associated with higher leverage, as the financial exposure of sectors is amplified,  is particularly strong for sector o,  that is the sector directly hit by productivity shocks.

\begin{table}
\caption{Sectors Profits, Directed Line}
\label{tab_profits_line}
\centering
{%
\begingroup\catcode`\_=12\relax
\begin{tabular}{|c|r|r|}
\hline
$\theta_k$ & $k$ & $\E[(\pi_k-\E[\pi_k])^2]^{1/2}$ \\ \hline
\hline
0.5 & 1 & 0 \\ \hline
0.5 & 2 & 0.82 \\ \hline
0.5 & 3 & 0.14 \\ \hline
0.5 & 4 & 0.06 \\ \hline
1 & 1 & 0 \\ \hline
1 & 2 & 0.88 \\ \hline
1 & 3 & 0.15 \\ \hline
1 & 4 & 0.07 \\ \hline
\hline
\end{tabular}
 
\endgroup
}
\end{table}






Finally, for clarity, we report in Table \ref{tab_cascade_line} the default probabilities. Because the chain is directed, and the shock originates only at node $o=2$,  default never occurs at node $1$ and cannot occur at a downstream node $k+1$ unless it occurs at the upstream node $k$. Thus, the marginal default probabilities coincide with the probabilities that the cascade up to that node occurs. 

\begin{table}
\caption{Cascade Effects, Directed line}
\label{tab_cascade_line}
\centering
{%
\begingroup\catcode`\_=12\relax
\begin{tabular}{|c|r|r|r|r|}
\hline
$k$ & 1 & 2 & 3 & 4 \\ \hline
\hline
$\P(\pi_k<0)$ & 0 & 0.6685 & 0.4641 & 0.3742 \\ \hline
\hline
\end{tabular}

\endgroup
}
\end{table}

The probability that sector~o defaults is 66.85\%, establishing the baseline likelihood of an initial failure. Joint default of sectors~o and~2 occurs in 46.41\% of the cases. Furthermore, the probability of a three-sector cascade involving sectors $o=2$, $3$, and $4$ is $37.42\%$. This shows that once the shock reaches sector $3$, it frequently continues to propagate further along the chain. The monotonic decline of these probabilities as additional sectors are included reflects the sequential attenuation of the shock: defaults become progressively less likely as the distance from the initially shocked sector increases. The values are large enough to indicate substantial systemic vulnerability and a strong cascade mechanism within the directed structure.


\subsection{Directed Cycle}
If in the previous example a link is added from sector $n$ to sector $1$ (with weight $A_{n1}=\alpha$), we obtain a production network forming a directed cycle, whose Leontief inverse matrix is 
\be
L =\frac{1}{1-\alpha^4}
\begin{bmatrix}
1 & \alpha^3 & \alpha^2 & \alpha\\
\alpha & 1 & \alpha^3 & \alpha^2\\
\alpha^2 & \alpha & 1 & \alpha^3\\
\alpha^3 & \alpha^2 & \alpha & 1
\end{bmatrix}.
\ee

\begin{center}
\begin{tikzpicture}[
    >=stealth,
    node/.style={circle, draw=gray!70, fill=gray!20, minimum size=10mm, inner sep=0pt}
]

    \def\r{2}

    \node[node]                              (n1) at (90:\r) {$1$};
    \node[circle, draw=red, fill=red!20,
          minimum size=10mm, inner sep=0pt]  (no) at (0:\r)  {$o$};
    \node[node]                              (n3) at (270:\r) {$3$};
    \node[node]                              (n4) at (180:\r) {$4$};

    \draw[->] (n1) to[bend left=20] (no);
    \draw[->] (no) to[bend left=20] (n3);
    \draw[->] (n3) to[bend left=20] (n4);
    \draw[->] (n4) to[bend left=20] (n1);

\end{tikzpicture}
\end{center}

Assuming that once again the primitive log-productivity shock is concentrated in a single node $o$, i.e., $\eta_k=0$ for every $k\ne o$, we obtain that  
$$\rho_k=L_{ko}\eta_o=\left\{\ba{lcl}\ds\frac{\alpha^{k-o}}{1-\alpha^{4}}\eta_o&\se&k\ge o\\[10pt]\ds\frac{\alpha^{4+k-o}}{1-\alpha^{4}}\eta_o&\se&k<o \,.\ea\right.$$
The main feature of the cycle is that total shocks $\rho_k$ can  be larger in magnitude than the primitive log-productivity shock $\eta_{o}$. 

For comparability with the previous  numerical example, we consider a directed cycle with four nodes, assuming that $\alpha=0.6$, and  $\gamma_k=1/4$ for every sector $k$, so that the Bonacich centrality is $$v_k^0=\frac{1+\alpha+\alpha^2+\alpha^3}{4(1-\alpha^4)}=\frac1{4(1-\alpha)}\,,\qquad \forall k=1,2,3,4\,.$$  

As in the previous case, we assume that  $-\eta_o$ is  exponentially distributed with parameter $\lambda=0.25$. On the real side, the equilibrium is now characterized by Table
\ref{tab_nominal_cycle}. The effects on the real quantities visible here are again consistent with Proposition \ref{prop:lev-effect-qp}. They are much more nuanced than along the line, since now every sector is a (direct or indirect) supplier of every other sector.

\begin{table}
\caption{Optimal Quantities in Equilibrium, Directed Cycle}
\label{tab_nominal_cycle}
\centering
{%
\begingroup\catcode`\_=12\relax
\begin{tabular}{|c|r|r|r|r|r|}
\hline
$\theta_k$ & $k$ & $y_k$ & $l_k$ & $c_k$ & $p_k/w_k$ \\ \hline
\hline
0.5 & 1 & 0.51 & 0.25 & 0.28 & 2.54 \\ \hline
0.5 & 2 & 0.43 & 0.2 & 0.2 & 9.87 \\ \hline
0.5 & 3 & 0.49 & 0.28 & 0.23 & 5.82 \\ \hline
0.5 & 4 & 0.52 & 0.27 & 0.25 & 3.7 \\ \hline
1 & 1 & 0.48 & 0.25 & 0.28 & 2.78 \\ \hline
1 & 2 & 0.37 & 0.18 & 0.18 & 12.21 \\ \hline
1 & 3 & 0.44 & 0.28 & 0.21 & 6.96 \\ \hline
1 & 4 & 0.49 & 0.28 & 0.25 & 4.23 \\ \hline
\hline
\end{tabular}

\endgroup
}
\end{table}
 We show the results for  the financial side in Table \ref{tab_financial_cycle}. The cost of debt to sector 1 is now positive, since it may be reached by default. In this particular example, the circular symmetry of the economy matrix $A$ and the fact that  $\gamma$ is uniform, make all sectors equally central according to the traditional Bonacich measure $v_0$. The levered measure though is not equal for all sectors, since they can be hit differently by default, and therefore have different marginal, total cost of debt and levered centrality.
 
\begin{table}
\caption{Financial Variables, Directed Cycle}
\label{tab_financial_cycle}
\centering
{%
\begingroup\catcode`\_=12\relax
\begin{tabular}{|c|r|r|r|r|r|}
\hline
$\theta_k$ & $k$ & $\zeta_k$ & $\xi_k$ & $v^{0}_k$ & $v^{\zeta}_k$ \\ \hline
\hline
0.5 & 1 & 0.04 & 0.24 & 0.625 & 0.66 \\ \hline
0.5 & 2 & 0.42 & 0.57 & 0.625 & 0.77 \\ \hline
0.5 & 3 & 0.1 & 0.44 & 0.625 & 0.76 \\ \hline
0.5 & 4 & 0.07 & 0.33 & 0.625 & 0.73 \\ \hline
1 & 1 & 0.06 & 0.34 & 0.625 & 0.67 \\ \hline
1 & 2 & 0.58 & 0.78 & 0.625 & 0.82 \\ \hline
1 & 3 & 0.15 & 0.62 & 0.625 & 0.82 \\ \hline
1 & 4 & 0.1 & 0.47 & 0.625 & 0.78 \\ \hline
\hline
\end{tabular}

\endgroup
}
\end{table}

As far as actual realizations are concerned, Figure~\ref{fig:tau_epsilon_cycle} plots the overlapped histogram of total and suppliers' shocks.

\begin{figure}
\centering
\includegraphics[width=1\textwidth]{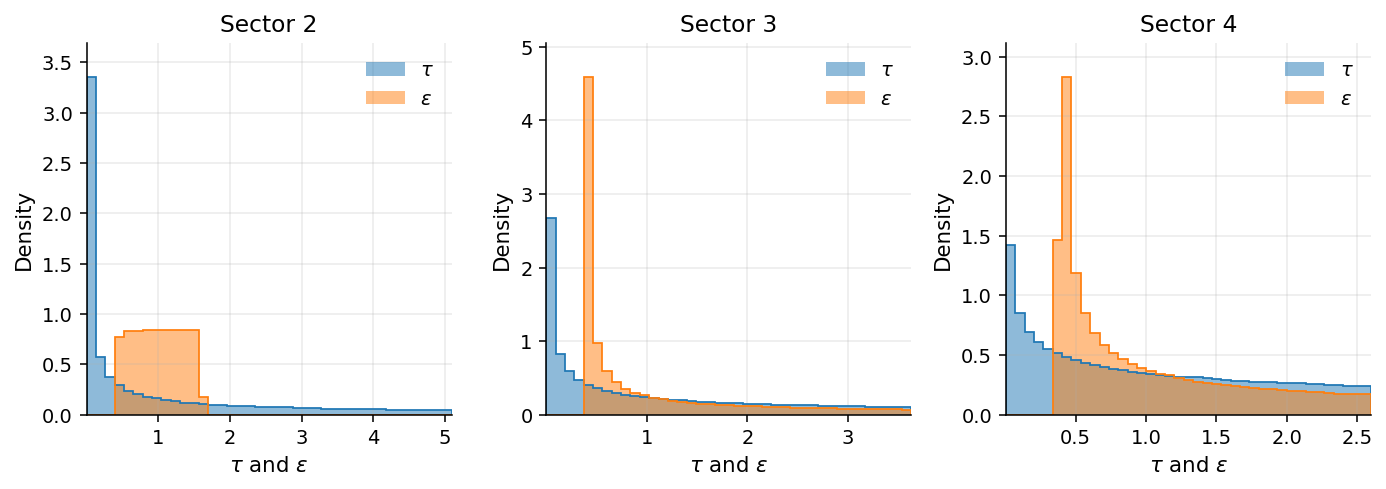}
\caption{Histogram of $\eps$ and $\tau$ by sector, $M=1000000$ draws, Directed Cycle.}
\label{fig:tau_epsilon_cycle}
\end{figure}

As expected in a feedback loop, $\tau_k$ tends to exceed $\eps_k$ for all sectors, reflecting amplification through circular propagation. 
Sector~o directly receives the primitive shock, while sectors~3 and~4 ---but also sector 1--- also show large spillovers due to network feedback.
The histogram suggests that the cycle structure raises both persistence and dispersion of shock transmission across the network.

The effect on profits in represented in the following Figure \ref{fig:profits_cycle} and Table \ref{tab:profits_cycle}. 

\begin{figure}
\centering
\includegraphics[width=1\textwidth]{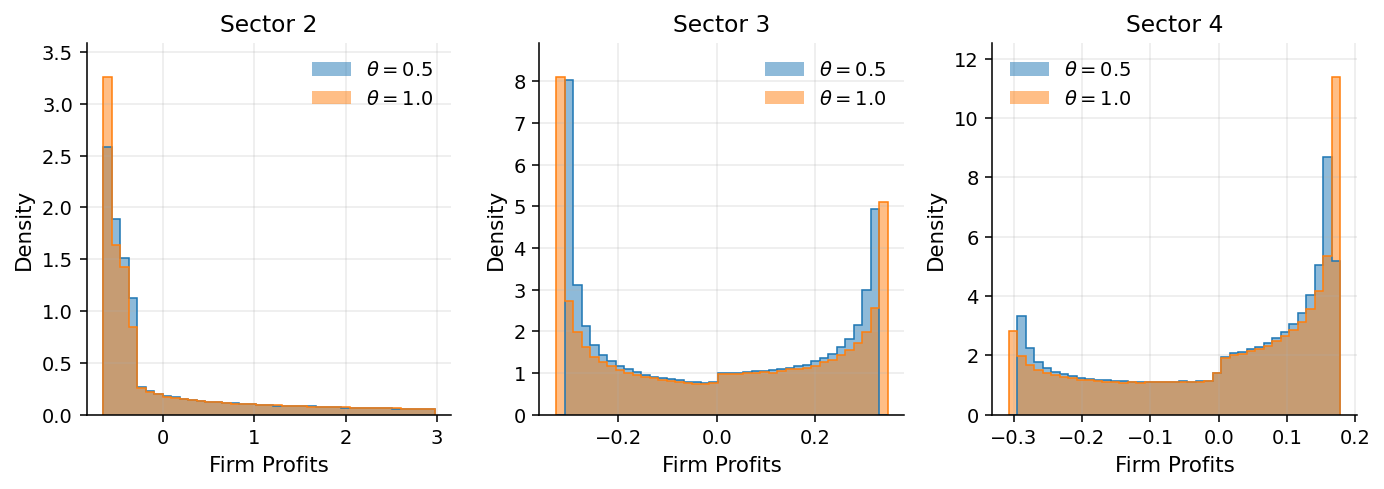}
\caption{Histogram of Sectors Profits, $M=1000000$ draws, Directed Cycle.}
\label{fig:profits_cycle}
\end{figure}

\begin{table}
\caption{Sectors Profits, Directed Cycle}
\label{tab:profits_cycle}
\centering
{%
\begingroup\catcode`\_=12\relax
\begin{tabular}{|c|r|r|}
\hline
$\theta_k$ & $k$ & $\E[(\pi_k-\E[\pi_k])^2])]^{1/2}$ \\ \hline
\hline
0.5 & 1 & 0.09 \\ \hline
0.5 & 2 & 0.89 \\ \hline
0.5 & 3 & 0.23 \\ \hline
0.5 & 4 & 0.15 \\ \hline
1 & 1 & 0.09 \\ \hline
1 & 2 & 0.95 \\ \hline
1 & 3 & 0.25 \\ \hline
1 & 4 & 0.16 \\ \hline
\hline
\end{tabular}

\endgroup
}
\end{table}

We present the marginal default probabilities in Table \ref{tab:cascade_cycle}. In the directed cycle, the probability that sector 1 defaults becomes positive (9\%). The probability that sector~o defaults is 72.41\%, which is substantially higher than in the equally parametrized line and reflects the stronger amplification created by the feedback loop. Also the default probabilities of the remaining sectors are higher than in the line.

\begin{table}
\caption{Cascade Effects, Directed Cycle}
\label{tab:cascade_cycle}
\centering
{%
\begingroup\catcode`\_=12\relax
\begin{tabular}{|c|r|r|r|r|}
\hline
$k$ & 1 & 2 & 3 & 4 \\ \hline
\hline
$\P(\pi_k<0)$& 31.79 & 72.41 & 49.01 & 39.74 \\ \hline
\hline
\end{tabular}

\endgroup
}
\end{table}
 Unlike the line case, even if the log-productivity shock can  occur only at node $o=2$,  default can occur at node $1$. However, it cannot occur downstream unless it occurs upstream. The cascade probabilities still coincide with the marginal default probabilities in the previous Table, whose  entries are bigger than in the line case. The joint failure of sectors~2 and~3 occurs in 49.01\% of the cases, indicating that the default propagates to the next sector almost half the times. The probability of a three-sector cascade involving sectors~2, 3, and~4 is 39.74\%. The decrease from three- to two-sector cascades shows that propagation attenuates with distance, but the levels remain high because the cycle architecture reinforces shocks by allowing them to circulate rather than dissipate. Overall, the magnitudes of these probabilities point to the fact  that the cycle network generates a significantly stronger cascade mechanism than the directed line, with the initial default being much more likely to trigger  failures.
 Last, let us recall that, in case the shock has an exponential distribution, the default conditions for nodes $k\geq o$ in a circle - given that $o$ has been hit - can be written as:
\begin{equation}\label{cycle1}
\begin{array}{ll}e^{\ell_k\eta_o}\frac{\ell_k+\lambda}{\lambda}-(1-\alpha)-\alpha e^{\ell_{k-1}\eta_o}\frac{\ell_{k-1}+\lambda}{\lambda}<0,\qquad &k>1,\\[5pt]
e^{\ell_0\eta_o}\frac{\ell_0+\lambda}{\lambda}-(1-\alpha)-\alpha e^{\ell_{n}\eta_o}\frac{\ell_{n}+\lambda}{\lambda}<0,\qquad &k=1,\end{array}
\end{equation}
where $$\ell_k=\frac{\alpha^k}{1-\alpha^{n+1}}, \qquad k=1,\dots , n.$$

\subsection{Directed Line vs Directed Cycle}

We next compare the directed line more closely with the directed cycle, to quantify how the presence of a feedback loop influences amplification. The comparison is conducted at $\theta = 1$.

\begin{center}
\begin{tabular}{|c|c|c|c|c|c|}
\hline
\multicolumn{6}{|c|}{\textbf{Directed Cycle $-$ Directed Line (Same $\theta=1$)}} 
\\\hline
$k$
  & $\Delta \zeta_k$ 
  & $\Delta \xi_k$ 
  & $\Delta y_k$ 
  & $\Delta c_k$ 
  & $\Delta \P(\pi_k<0)$ 
\\\hline
1 & 0.06  & 0.34 & $-0.07$ & $-0.08$ & $0.3179$ \\\hline
2 & $-0.19$ & 0.01 & $+0.06$ & $+0.01$ & $+0.0556$ \\\hline
3 & 0.02  & 0.03 & $+0.13$ & $+0.01$ & $0.026$ \\\hline
4 & 0.02  & 0.03 & $+0.26$ & $+0.02$ & $+0.0233$ \\\hline
\end{tabular}
\end{center}

The comparison  shows clear quantitative differences across all sectors. For sector $1$, the cycle produces higher cost of debt ($\Delta \zeta_1 = 0.06$, $\Delta \xi_1 = 0.34$) together with lower output and consumption ($\Delta y_1 = -0.07$, $\Delta c_1 = -0.08$) and a much higher probability of default ($\Delta \text{Default}_o = 31.79$ percentage points). This indicates that, relative to the line, the cycle substantially increases the exposure of the sector. In sector~o, directly hit by the shock, the cycle lowers $\zeta$ ($\Delta \zeta_o = -0.19$) while leaving $\xi$ almost unchanged ($\Delta \xi_o = 0.01$), and it yields slightly higher real activity ($\Delta y_o = 0.06$, $\Delta c_o = 0.01$) and a higher default rate ($\Delta \text{Default}_o = 5.56$). The decrease in the marginal cost of debt, together with the stability of the total one and the behavior of the sector default probability,  points to a more distributed default structure in the cycle. In sectors $3$ and $4$, cost of debt effects are small ($\Delta \zeta_3 = 0.02$, $\Delta \xi_3 = 0.03$, with the same order of magnitude for sector $4$). Production and consumption increase in both sectors, as in sector $o$ ($\Delta y_3 = 0.13$, $\Delta y_4 = 0.26$, $\Delta c_3 = 0.01$, $\Delta c_4 = 0.02$), and default rates rise by $2.60$ and $2.33$ percentage points, respectively. Overall, the sign and magnitude of these differences show that the directed cycle transmits and amplifies shocks more strongly than the directed line: the upstream sector in the line becomes involved in the cycle loop and experiences markedly lower real outcomes and positive marginal, higher total cost of debt, while downstream sectors display uniformly higher real activity and moderately higher financial risk.

\section{Summary and conclusions} \label{sec:conclu}

 This paper studies the effects of production shocks on rigid networks, namely networks in which decisions on intermediate goods and output cannot be shock-contingent, taking into consideration reliance on external debt. 
 
 In this economy, the Modigliani-Miller irrelevance does not hold: the way in which sector operations are funded, whether through debt or equity, becomes relevant. Default of single or network-related sectors may occur when real shocks occur, even in equilibrium, and the banks that finance defaulting sectors are affected and recover less than due.
 
 With perfect information, our economy returns a standard first-best equilibrium {\it\`a la} Acemoglu et al.\ \cite{acemoglu2012network}.
 
 Without perfect information,  the impact of production shocks on welfare is still represented by the Bonacich centrality vector. However, the elements of this vector  no longer represent the Domar weights. As a consequence, Hulten's theorem no longer holds and shock effects can be much bigger than one's own GDP share. This rationalizes the empirical evidence we have witnessed over recent war and pandemic episodes and warns about the future.  
 
 If there is also leverage, the rigid equilibrium, as expected, is a second-best and therefore reduces welfare. Consumption and labor usage  of sectors who are   suppliers of a levered one increases with respect to the unlevered case, while production of sectors which are not its suppliers goes down, and prices go up.  The equilibrium with leverage is profoundly different from the unlevered one even if shocks do not occur, ex post.
 
 So, rigidity and default have real and inflationary consequences, consistently with recent macro evidence.

 The main feature of the equilibrium with rigidity and leverage is the possibility of default. We distinguish the frequency from the severity of default.  The occurrence of sectors' default just depends on the real shocks and the network structure,  the magnitude of the losses depends on the  interest rates paid to financiers, which are endogenous, and the network structure.  Regarding occurrence,  sector $k$ is in default if its productive perturbation is strictly less than the sum of the perturbations on the neighboring nodes, weighted by their  importance in the production process. The shocks that enter this condition  already comprehend their network effects. 
 The magnitude of the losses also depends on the Bonacich centrality adjusted for debt.
 We deem these results quite important in the current context, to have a sense of how many default cascades disruptions of supply chains and other production shocks can generate and what circumstances can soothe or worsen them.
 
 The numerical examples give a sense of the richness of effects that the model permits one to study, once faced even with simple networks, and how much it can tell about the effects on real as well as financial variables. They allow us to realize how much of the interplay between the real and financial side we can investigate.

  Our analytical results could be used for designing policy interventions on the structure of the network or on its reliance on leverage to control inflation,  or consumption or labor demand in specific sectors, or to mitigate default probabilities. Thanks to the insights of the model, specific policy measures could be adopted to  address the consequences on maximal equilibrium quantities and prices, as well as on  default propagation and loss amplification.  Equilibrium prices and quantities, as well as default occurrence and propagation,  depend on  the input-output structure of the production network, and interventions in that sense can be designed based on our analytical results.  
 
 There is more than that, though: if the network structure cannot be changed, at least in the short run, our model suggests that effects on maximal quantities and prices, but also amplification and default cascades,  can still be controlled working on the financial side. Exogenous policy interventions that limit leverage or the cost of debt can be designed to prevent welfare impacts or loss cascades in case shocks hit. Changes in the network are probably more difficult and lengthy to implement in practice than a reduction in leverage of the most risky or most interconnected sectors.

 
 An attempt to design these short-run policy instruments, coupled with a calibration of the model to a specific country input-output matrix, is therefore in our agenda.

\begin{appendix}

\section{Appendix}
\subsection{Proof of Proposition \ref{prop:actual}}\label{sec:proof-proposition1}
For every sector $k\in\mc V$, define the function 
$$F_k(z,l)=\varsigma_kl_k^{\beta_k}\prod \limits_{j\in\mc V}z_{jk}^{A_{jk}}\,,$$
so that the Cobb-Douglas production relation \eqref{Cobb-Douglas-1} can equivalently be rewritten as \be\label{Cobb-Douglas-1bis}y_k^{\eta}=e^{\eta_k}F_k(z^{\eta},l)\,.\ee
for every industry $k\in\mc V$. 

(i) Every quadruple $(\tilde y,\tilde z,\tilde c,l)$  satisfying the proportional rationing rule \eqref{proportional-rationing-1} must satisfy Equation\eqref{actual-quantities} for some $\rho$. Moreover, if $(\tilde y,\tilde z,\tilde c,l)$ satisfies the Cobb-Douglas production relation \eqref{Cobb-Douglas-1bis} for every realization of the log-productivity shock vector $\eta$, then 
$$e^{\rho_k}y^0_k=y^\eta_k=e^{\eta_k}F_k(z^{\eta},l)=e^{\eta_k+\sum\limits_{j\in\mc V}\rho_jA_{jk}}F_k(z,l)=e^{\eta_k+\sum\limits_{j\in\mc V}\rho_jA_{jk}}y_k^0\,,$$
for every sector $k\in\mc V$, so that necessarily
$\rho=\eta+A'\rho$, which is equivalent to Equation \eqref{rho-def}. 

(i) For the quadruple $(y^{\eta}, z^{\eta},  c^{\eta},l)$,  Equation\eqref{actual-quantities} clearly implies the proportional rationing rule \eqref{proportional-rationing-1}-\eqref{proportional-rationing-2}. Now, observe that Equation \eqref{rho} implies that 
$$\eta+A'\rho=\eta+\rho-(I-A')\rho=\rho\,.$$
Then, the proportional rationing rule \eqref{proportional-rationing-1} implies that 
$$\frac{e^{\eta_k}F_k(z^{\eta},l)}{y^\eta_k}=\frac{e^{\eta_k}\varsigma_kl_k^{\beta_k}\prod \limits_{j\in\mc V}(z^0_{jk}e^{\rho_j})^{A_{jk}}}{e^{\rho_k}y^0_k}=\frac{e^{\eta_k+\sum_{j\in\mc V}\rho_jA_{jk}}F_k(z^0,l)}{e^{\rho_k}y^0_k}=\frac{F_k(z^0,l)}{y^0_k}\,,$$
for every industry $k\in\mc V$, from which it is clear that the Cobb-Douglas production relation \eqref{Cobb-Douglas-1bis} holds for every realization of the log-productivity shock vector $\eta$ if and only if it holds for $\eta=0$. 
Similarly, the proportional rationing rule \eqref{proportional-rationing-1} implies that 
$$y_k^\eta-\sum\limits_{j\in\mc V}z^\eta_{kj}-c^\eta_k=
e^{\rho_k}\left(y_k^0-\sum\limits_{j\in\mc V}z^0_{kj}-c^0_k\right)\,,$$
so that the market for goods clearing condition \eqref{market-clearing} is satisfied for every realization of the log-productivity shock vector $\eta$ if and only if it is satisfied for $\eta=0$. 
\qed

\subsection{Proof of Proposition \ref{prop-zeta}}\label{sec:prop-zeta}
Throughout this proof, it is convenient to adopt the compact notation \be\label{clump-notation}[a]_b^c:=\min\{\max\{a,b\},c\}\,.\ee
Consider the function $f:[1,+\infty)\times[0,1]\to\R$ defined by 
$$f(x,\theta)=\E\left[[x\tau_k]_{(1-\theta)\eps_k}^{x\eps_k}|\varphi(\eta)\right]-1
\,,$$ 
and observe that Equation \eqref{distortion-def} is equivalent to \be\label{distortion-def-bis} f(e^{\zeta_k},\theta_k)=0\,.\ee
For every fixed $\theta$ in $[0,1]$, Equation \eqref{Eepsk=1} implies that 
$$f(1,\theta)=\E\left[[\tau_k]_{(1-\theta)\eps_k}^{\eps_k}|\varphi(\eta)\right]-1
\le\E[\eps_k|\varphi(\eta)]-1=0\,,$$
whereas, since $\min\{\eps_k,\tau_k\}>0$, we have that 
$$f(x,\theta)\ge x\E[\min\{\tau_k,\eps_k\}|\varphi(\eta)]-1\stackrel{{x\to+\infty}}{\longrightarrow}+\infty\,.$$
Since $f(x,\theta)$ is continuous and non-decreasing in $x$,  it follows that, for every $\theta$ in $[0,1]$, there exists $x^*(\theta)\ge1$ such that 
$f(x,\theta)<0$ for every  $x$ in $[1,x^*(\theta))$, $f(x^*(\theta),\theta)=0$, and $f(x,\theta)\ge0$ for every  $x\ge x^*(\theta)$. In particular, $\zeta_k=\log x^*(\theta_k)$ is a solution of Equation \eqref{distortion-def-bis}. 
Since $f(x,\theta)$ is non-increasing in $\theta$, we get that $x^*(\theta)$ is non-decreasing in $\theta$, hence so is $\zeta_k$ as a function of $\theta_k$.

We now prove uniqueness. Towards this goal, observe that Equations \eqref{Etauk=1} and \eqref{Eepsk=1} imply $\E[\tau_k-\eps_k|\varphi(\eta)]=0$, so that  
$\P(\tau_k\ge\eps_k|\varphi(\eta))>0$.
Since $\eps_k>0$, this implies that 
$$\E\left[\eps_k\1_{\{\tau_k\ge\eps_k\}}|\varphi(\eta)\right]>0\,.$$
Hence, for $\theta$ in $[0,1]$ and $x_1>x_0\ge 1$, we have 
$$\ba{rcl}f(x_1,\theta)
&=&\E\left[\max\left\{x_1\min\{\tau_k,\eps_k\right\},(1-\theta)\eps_k\}|\varphi(\eta)\right]-1\\ 
&=&x_1\E\left[\eps_k\1_{\{\tau_k\ge\eps_k\}}|\varphi(\eta)\right]-1\\
&&+\E\left[\max\left\{x_1\min\{\tau_k,\eps_k\right\},(1-\theta)\eps_k\}\1_{\{\tau_k<\eps_k\}}|\varphi(\eta)\right]\\
&>&x_0\E\left[\eps_k\1_{\{\tau_k\ge\eps_k\}}|\varphi(\eta)\right]-1\\
&&+\E\left[\max\left\{x_0\min\{\tau_k,\eps_k\right\}\1_{\{\tau_k<\eps_k\}}|\varphi(\eta)\right]\\
&=&\E\left[\max\left\{x_0\min\{\tau_k,\eps_k\right\},(1-\theta)\eps_k\}|\varphi(\eta)\right]-1\\
&=&f(x_0,\theta)
\,,\ea$$ 
so that $f(x,\theta)$ is strictly increasing in $x\ge1$. In particular, this implies that $f(x,\theta)>0$ for every $x>x^*(\theta)$.  Therefore, Equation \eqref{distortion-def-bis} has a unique solution $\zeta_k=\log x^*(\theta_k)$. 

Observe that 
$$f(1,0)=\E\left[\eps_k|\varphi(\eta)\right]-1=0\,,$$
so that $x^*(0)=1$, which implies that $\zeta_k(0,\varphi(\eta))=\log x^*(0)=0$, proving Equation \eqref{zeta-min}. On the other hand, we have that 
$$f(x,1)=x\E\left[\min\{\eps_k,\tau_k\}|\varphi(\eta)\right]-1\,,$$ whose zero is $x^*(1)=1/\E[\min\{\eps_k,\tau_k\}|\varphi(\eta)]$. This implies that $\zeta_k(1,\varphi(\eta))=\log x^*(1)=-\log\E[\min\{\eps_k,\tau_k\}|\varphi(\eta)]$, thus proving Equation \eqref{zeta-max}.

Finally, recall that, in the special case of full information, i.e., when $\varphi(\eta)=\eta$, the variables $\tau_k$ and $\eps_k$ are identically equal to $1$, so that $f(x,\theta)=\max\{x,1-\theta\}-1=x-1$ for every $x\ge1$ and $\theta$ in $[0,1]$. Hence, in this case $x^*(\theta)=1$, so that $\zeta_k(\theta_k,\eta)=0$ for every $\theta_k$ in $[0,1]$, thus proving Equation \eqref{zeta=0}.  
\qed

\subsection{Proof of Lemma \ref{lemma:monotonicity-zeta}}
\label{sec:proof-lemma-monotonicity-zeta}

First, notice that $e^{-[\zeta]}A'$ is a matrix whose entries $(e^{-[\zeta]}A')_{jk}=e^{-\zeta_j}A_{kj}$ are non-increasing in $\zeta$ for every two sectors $j$ and $k\in\mc V$, and strictly decreasing in $\zeta_j$ if and only if  $A_{kj}>0$, i.e., industry $k$ is a direct supplier of industry $j$. It follows that, for every two sectors $j$ and $k\in\mc V$, the corresponding entry 
$$(L^{\zeta})_{jk}=1+(e^{-[\zeta]}A')_{jk}+(e^{-[\zeta]}A'e^{-[\zeta]}A')_{jk}+\ldots\,  ,$$
of the discounted Leontief matrix $L^\zeta$ is monotone non-increasing in $\zeta$ thus proving point (i). In fact, we see that $(L^{\zeta})_{jk}$ is strictly decreasing in $\zeta_o$, for a sector $o$ (including possibly node $j$ itself) if and only if, in the directed graph describing the production network, there exists a path from node $k$ to node $j$ passing through node $o$, i.e., if and only if $L_{jo}L_{Lok}>0$. In particular, if sector $k$ is not a supplier of sector $o$, i.e., if $L_{ok}=0$, then $(A^h)_{ko}=0$ for every $h\ge1$, so that $L_{jk}^{\zeta}$ is constant in $\zeta_o$ for every sector $j\in\mc V$, thus proving the first part of point (iv). 

Analogously, the normalization quantity 
$$\psi(\zeta):=\sum_{j\in\mc V}\sum_{k\in\mc V}\gamma_jL_{jk}^{\zeta}\beta_ke^{-\zeta_k}\,,$$ is a monotone non-increasing function of $\zeta$ (proving point (ii)). In particular, a strict increase in the primitive cost of debt $\zeta_o$ in any sector $o$ such that either $\gamma_o>0$ or $\gamma_o=0$ and $o$ is a supplier of some sector $j$ with $\gamma_j>0$ causes a strict decrease of $\psi(\zeta)$, thus proving point (v). 

To prove point (iii), simply observe that the non-normalized discounted centrality of a sector $k$
$$\psi(\zeta)v_k^{\zeta}=\sum_{j\in\mc V}L_{jk}^{\zeta}\gamma_j\,,$$
is monotone non-decreasing in $L^{\zeta}$, which in turn is monotone non-increasing in $\zeta$; . 

Finally, consider the case when sector $k$ is not a supplier of sector $o$, so that $L_{kj}^{\zeta}$ is constant in $\zeta_o$ for every $j$, hence so is the non-normalized discounted centrality $\psi(\zeta)v_k^{\zeta}=\sum_{j\in\mc V}L_{jk}^{\zeta}\gamma_j$. Since the normalization constant $\psi(\zeta)$ is monotone non-increasing in $\zeta$ by point (iii), we have that the normalized discounted centrality $v_k^{\zeta}$ is monotone non-decreasing in $\zeta_o$, thus proving the second part of point (iv). In this case, if additionally either $\gamma_o>0$ or $\gamma_o=0$ and sector $o$ is a supplier of some sector $j\ne k$ with $\gamma_j>0$,  we have that the normalization constant  $\psi(\zeta)$ is strictly increasing in $\zeta_o$, so that the normalized discounted centrality $v_k^{\zeta}$ is strictly increasing in $\zeta_o$, thus proving point (vi). \qed

\subsection{Proof of Theorem \ref{theo:Walrasian} and Corollary \ref{coro:Walrasian}}\label{sec:proof-theorem1}

 We first prove that given any $\varphi$-rigid Walrasian equilibrium $(y^0,z^0,l,c^0,r, p,w)$, all relations \eqref{max-y}, \eqref{max-z}, \eqref{labor}, \eqref{max-c}, \eqref{rkopt}, \eqref{prices}, \eqref{wage}, and \eqref{price-consbundle} of Theorem \ref{theo:Walrasian}, as well relations \eqref{actual-y}, \eqref{actual-z}, \eqref{actual-c}, \eqref{actual-profit}, \eqref{actual-budget}, and \eqref{actual-utility} of Corollary \ref{coro:Walrasian} hold true.

For every sector $k\in\mc V$, let
$$s_k=\E[\mc A_k(\eta)|\varphi(\eta)]\,,$$ 
be the conditional expected values of industry $k$'s assets given the public signal. 
Equation \eqref{assets} and the first identity in Equation \eqref{actual-quantities} 
imply that  
\be\label{sk}s_k=\E[\mc A_k(\eta)|\varphi(\eta)]
=\E[p_ky^\eta_{k}|\varphi(\eta)]
=p_ky_k^0\E[e^{\rho_k}|\varphi(\eta)]\,,\qquad \forall k\in\mc V\,.\ee
Using again the first identity in Equation \eqref{actual-quantities}, we get that the industries' actual assets are \be\label{actual-Ak}\mc A_k(\eta)=
p_ky^\eta_k=p_ky^0_ke^{\rho_k}=p_ky^0_k\E[e^{\rho_k}|\varphi(\eta)]\tau_k=s_k\tau_k\,,\ee
where $\tau_k$ is the normalized total shock defined in Equation \eqref{tau-def}. 

On the other hand, Equation \eqref{liabilities} and the second identity in \eqref{actual-quantities} allow us to write the conditional expected value of the industries' liabilities given the public signal as
$$\E[\mc L_k(\eta)|\varphi(\eta)]=\E\left[\left(\sum_{j\in\mc V}p_jz_{jk}^\eta+wl_k\right)|\varphi(\eta)\right]=
\sum_{j\in\mc V}p_jz_{jk}^0\E[e^{\rho_j}|\varphi(\eta)]+wl_k\,,
$$
for every sector $k\in\mc V$. 
Now, the Cobb-Douglas form of the production function \be\label{Cobb-Douglas-0}y_k^0=\varsigma_kl_k^{\beta_k}\prod_{j\in\mc V}(z^0_{jk})^{A_{jk}}\,,\ee
implies that 
$$\frac{\partial y_k^0}{\partial z^0_{jk}}=\frac{A_{jk}y^0_k}{z_{jk}^0}\,,\qquad\frac{\partial y_k^0}{\partial l_{k}}=\frac{\beta_{k}y^0_k}{l_{k}}\,.\qquad \forall j,k\in\mc V\,.$$
It follows that the industries' conditional expected profits given the public signal 
\begin{align*}
\E[\pi_k(\eta)|\varphi(\eta)]
&=\ds\E[\mc A_k(\eta)-(1+r_k\theta_k)\mc L_k(\eta)|\varphi(\eta)] \\
&=\ds p_ky_k^0
\E[e^{\rho_k}|\varphi(\eta)]-(1+r_k\theta_k)\left(\sum_{j\in\mc V}p_jz_{jk}^0\E[e^{\rho_j}|\varphi(\eta)]+wl_k\right)
\end{align*}
are differentiable concave functions of the decision variables $(z^0_{jk})_{j\in\mc V}$ and $l_k$, with first-order necessary and sufficient conditions for their maximization given by  
$$0=\frac{\partial}{\partial z^0_{jk}}\E[\pi_k(\eta)|\varphi(\eta)]=p_k\frac{A_{jk}y^0_k}{z_{jk}^0}\E[e^{\rho_k}|\varphi(\eta)]-(1+r_k\theta_k)p_j\E[e^{\rho_j}|\varphi(\eta)]\,,$$
that using \eqref{sk} can be rewritten as
\be\label{FOC-z}
p_jz^0_{jk}\E[e^{\rho_j}|\varphi(\eta)]=\frac{A_{jk}s_k}{1+r_k\theta_k}\,,\qquad \forall j,k\in\mc V\,,\ee
and, respectively,
$$0=\frac{\partial}{\partial l_k}\E[\pi_k(\eta)|\varphi(\eta)]
=p_k\frac{\beta_{k}y^0_k}{l_{k}}\E[e^{\rho_k}|\varphi(\eta)]-(1+r_k\theta_k)w\,, $$
that can be rewritten as
\be\label{FOC-l}
wl_k=\frac{\beta_ks_k}{1+r_k\theta_k}\,,\qquad \forall k\in\mc V\,.
\ee
It then follows from  Equation \eqref{liabilities}, the second identity in \eqref{actual-quantities}, and the first-order optimality conditions \eqref{FOC-z} and \eqref{FOC-l} that the actual liabilities of industry $k$ are \be\label{actual-Lk}
\ba{rcl}\ds\mc L_k(\eta)
&=&\ds\sum_{j\in\mc V}p_jz^{\eta}_{jk}+wl_k\\
&=&\ds\sum_{j\in\mc V}p_jz^{0}_{jk}e^{\rho_j}+wl_k\\
&=&\ds\sum_{j\in\mc V}\tau_j\frac {A_{jk}s_k}{1+r_k\theta_k}+\frac{\beta_ks_k}{1+r_k\theta_k}\\
&=&\ds\eps_k\frac{s_k}{1+r_k\theta_k}
\,,\ea\ee
where $\eps_k$ is the information-normalized total shock to the suppliers of sector $k$, as defined in Equation \eqref{epsk}. 

Now, observe that, by Equation \eqref{default-costs}, we can rewrite the default costs as 
$$\ba{rcl}\mc D_k(\eta)&=&
\left[(1+r_k)\theta_k\mc L_k(\eta) -[\mc A_k(\eta)-(1-\theta_k)\mc L_k(\eta)]_+\right]_+ \\
&=&\left[\min\{(1+r_k)\theta_k\mc L_k(\eta),(1+r_k\theta_k)\mc L_k(\eta) -\mc A_k(\eta)\}\right]_+\\
&=&\left[(1+r_k\theta_k)\mc L_k(\eta) -\mc A_k(\eta)\right]_0^{(1+r_k)\theta_k\mc L_k(\eta)}
\,,\ea$$
where we have used the notation \eqref{clump-notation}. 
By Equation \eqref{Bank-profit}, it follows that the banks' profits are given by 
$$
\ba{rcl}\mathcal I_k(\eta)
&=&r_k\theta_k\mc L_k(\eta)-\mc D_k(\eta)\\
&=&r_k\theta_k\mc L_k(\eta)-\left[(1+r_k\theta_k)\mc L_k(\eta) -\mc A_k(\eta)\right]_0^{(1+r_k)\theta_k\mc L_k(\eta)}\\
&=&r_k\theta_k\mc L_k(\eta)+\left[\mc A_k(\eta)-(1+r_k\theta_k)\mc L_k(\eta) \right]_{-(1+r_k)\theta_k\mc L_k(\eta)}^0\\
&=&\left[\mc A_k(\eta)-\mc L_k(\eta)\right]^{r_k\theta_k\mc L_k(\eta)}_{-\theta_k\mc L_k(\eta)}\ea$$
Substituting the above into the null conditional expected bank's profit condition \eqref{bank-optimality} yields  
\be\label{Emin}
\ba{rcl}0
&=&\ds\frac{1+r_k\theta_k}{s_k}
\E\left[\mc I_k(\eta)|\varphi(\eta)\right]\\[10pt]
&=&\ds\frac{1+r_k\theta_k}{s_k}
\E\left[\left[\mc A_k(\eta)-\mc L_k(\eta)\right]^{r_k\theta_k\mc L_k(\eta)}_{-\theta_k\mc L_k(\eta)}|\varphi(\eta)\right]\\[10pt]
&=&\E\left[\left[(1+r_k\theta_k)\tau_k-\eps_k\right]^{r_k\theta_k\eps_k}_{-\theta_k\eps_k}\big|\varphi(\eta)\right]\\[10pt]
&=&\E\left[\left[(1+r_k\theta_k)\tau_k\right]^{(1+r_k\theta_k)\eps_k}_{(1-\theta_k)\eps_k}\big|\varphi(\eta)\right]-\E[\eps_k|\varphi(\eta)]\\[10pt] 
&=&\E\left[\left[(1+r_k\theta_k)\tau_k\right]^{(1+r_k\theta_k)\eps_k}_{(1-\theta_k)\eps_k}\big|\varphi(\eta)\right]-1
\,,\ea\ee
where the last identity follows from Equation \eqref{Eepsk=1}. 
It then follows that 
\be\label{1+rk}1+r_k\theta_k=e^{\zeta_k}\,,\qquad \forall k\in\mc V\,,\ee
where $\zeta_k$ is the solution of Equation \eqref{distortion-def}, so that Equation \eqref{rkopt} holds true. 

Similarly, it follows from Equation \eqref{Cobb-Douglas-2} and the third identity in \eqref{actual-quantities} that the conditional expected value of the representative household's utility 
$$\E[U(c^{\eta})|\varphi(\eta)]=\E[\chi\prod_{k\in\mc V}(c_k^{\eta})^{\gamma_k}]=\chi\prod_{k\in\mc V}(c_k^{0})^{\gamma_k}\E[e^{\sum_{k\in\mc V}\gamma_k\rho_k}|\varphi(\eta)]\,,$$
is a concave function of $c^0$ and the first-order necessary and sufficient conditions for its maximization under the conditional expected budget constraint 
\be\label{expected-budget-constraint}\sum_{k\in\mc V}p_kc_k^0\E[e^{\rho_k}|\varphi(\eta)]=
\sum_{k\in\mc V}p_k\E[c_k^{\eta}|\varphi(\eta)]\le E\,,
\ee
read 
\be\label{FOC-c}
p_k\E[e^{\rho_k}|\varphi(\eta)]c^0_k=E\gamma_k\,,\qquad \forall k\in\mc V\,.
\ee
It then follows from Equation \eqref{sk}, the clearing condition \eqref{market-clearing} for the market for goods, the first-order optimality conditions \eqref{FOC-z} and \eqref{FOC-c}, and Equation \eqref{1+rk} that   
\be\label{s}
\ba{rcl}s_j
&=&\ds p_jy_j^0\E[e^{\rho_j}|\varphi(\eta)]\\
&=&\ds p_j\left(\sum_{k\in\mc V}z^0_{jk}+c_j^0\right)\E[e^{\rho_j}|\varphi(\eta)]\\
&=&\ds\sum_{k\in\mc V}\frac {A_{jk}}{1+r_k\theta_k}s_k+E\gamma_j\\
&=&\ds\sum_{k\in\mc V}{A_{jk}}e^{-\zeta_k}s_k+E\gamma_j\,,\ea\ee
for every $j\in\mc V$. 
Notice that Equation \eqref{s} is equivalent to 
\be\label{s-bis}s_k=E\sum_{j\in\mc V} L_{jk}^{\zeta}\gamma_j\,,\qquad \forall k\in\mc V\,,\ee
where $L^\zeta$ is the discounted Leontief matrix defined in \eqref{distorted-Leontief}. From the clearing condition \eqref{labor-clearing} for the labor market, the first-order optimality condition \eqref{FOC-l}, and Equations \eqref{1+rk}, \eqref{s-bis}, and \eqref{psi}, we get 
\be\label{w-E}w=w\sum_{k\in\mc V}l_k=\sum_{k\in\mc V}\frac{\beta_ks_k}{1+r_k\theta_k}=E\sum_{j\in\mc V}\sum_{k\in\mc V}\beta_ke^{-\zeta_k}L^{\zeta}_{jk}\gamma_j=E\psi(\zeta)\,.\ee
From Equations \eqref{s-bis} and \eqref{w-E}, we get 
\be\label{s-ter}s_k=E\sum_{j\in\mc V} L_{jk}^{\zeta}\gamma_j=\frac{w}{\psi(\zeta)}\sum_{j\in\mc V} L_{jk}^{\zeta}\gamma_j=wv^{\zeta}_k\,,\qquad \forall k\in\mc V\,,\ee
where 
$v^{\zeta}_k$ is the normalized discounted centrality of sector $k$, as defined in Equation
\eqref{centrality-def}.
It now follows from Equation \eqref{Cobb-Douglas-0} and the first-order optimality conditions \eqref{FOC-z} and \eqref{FOC-l} that 
\be\label{CB-more} \ba{rcl}
y_k^0
&=&\ds\varsigma_kl_k^{\beta_k}\prod_{j\in\mc V}(z^0_{jk})^{A_{jk}}\\
&=&\ds\varsigma_k\frac{\beta_k^{\beta_k}}{w^{\beta_k}}\left(\frac{s_k}{1+r_k\theta_k}\right)^{\beta_k+\sum_{j\in\mc V}A_{j_k}}\frac{\prod_{j\in\mc V}A_{jk}^{A_{jk}}}{\prod_{j\in\mc V}p_j^{A_{jk}}\prod_{j\in\mc V}\E[e^{\rho_j}|\varphi(\eta)]^{A_{jk}}}\\
&=&\ds\frac{y^0_k(p_k/w)\E[e^{\rho_k}|\varphi(\eta)]e^{-\zeta_k}}{\prod_{j\in\mc V}(p_j/w)^{A_{jk}}\prod_{j\in\mc V}\E[e^{\rho_j}|\varphi(\eta)]^{A_{jk}}}
\,,\ea\ee
where the last identity is a consequence of the choice of the normalization constant $\varsigma_k=\beta_k^{-\beta_k}\prod_{j\in\mc V}A_{jk}^{-A_{jk}}$, the constant returns-to-scale assumption \eqref{normalization-1} for the industries' production functions, and Equation \eqref{sk}. By taking the logarithm of both sides of Equation \eqref{CB-more} and rearranging terms, we get 
$$\log\frac{p_k\E[e^{\rho_k}|\varphi(\eta)]}{w}-\sum_{j\in\mc V}A_{jk}\log\frac{p_j\E[e^{\rho_j}|\varphi(\eta)]}{w}=\zeta_k\,,\qquad \forall k\in\mc V\,,$$ 
which is equivalent to 
$$\log\frac{p_k\E[e^{\rho_k}|\varphi(\eta)]}{w}=\sum_{j\in\mc V}L_{kj}\zeta_j=\xi_k\,,\qquad \forall k\in\mc V\,,$$
where $\xi_k$ is the total cost of debt for industry $k$, as defined in \eqref{xi-def}. Clearly, the above equation is equivalent to Equation \eqref{prices}. From Equations \eqref{sk}, \eqref{s-ter}, and \eqref{prices}, we then get 
$$y^0_k
=\frac{s_k}{p_k\E[e^{\rho_k}|\varphi(\eta)]}=\frac{wv_k^{\zeta}}{p_k\E[e^{\rho_k}|\varphi(\eta)]}=v_k^{\zeta}e^{-\xi_k}\,,\qquad\forall k\in\mc V\,,$$
thus proving Equation \eqref{max-y}. 

Equation \eqref{prices} is equivalent to  
\be\label{prices-equiv}\log w=\log p_k+\log\E[e^{\rho_k}]-\xi_k\,,\ee
for every sector $k\in\mc V$. 
Now, recall that, from the definitions \eqref{xi-def} of the total cost of debt and \eqref{centrality-def} of the discounted centrality, it follows that \be\label{sumgammaxi}\sum_{k\in\mc V}\gamma_k\xi_k= \sum_{j\in\mc V}\sum_{k\in\mc V}\gamma_jL_{jk}\zeta_k= \sum_{k\in\mc V}v_k^0\zeta_k\,.\ee
Multiplying both sides of the Equation \eqref{prices-equiv} by $\gamma_k$, summing up over $k\in\mc V$, using the constant returns-to-scale assumption \eqref{normalization-2}, and substituting Equation \eqref{sumgammaxi}, we get that 
$$\ba{rcl}\log w&=&\ds\sum_{k\in\mc V}\gamma_k\log w\\ &=&\ds\sum_{k\in\mc V}\gamma_k\log p_k+\ds\sum_{k\in\mc V}\gamma_k\log\E[e^{\rho_k}]-\ds\sum_{k\in\mc V}\gamma_k\xi_k\\ &=&\ds\sum_{k\in\mc V}\gamma_k\log p_k+\ds\sum_{k\in\mc V}\gamma_k\log\E[e^{\rho_k}]-\ds\sum_{k\in\mc V}v^0_k\zeta_k\,,\ea$$
that is equivalent to Equation \eqref{wage}. It follows from Equations \eqref{prices} and \eqref{wage} of Theorem \ref{theo:Walrasian} that the unit price $p_k$ of $k$-th good normalized by the one $\prod_jp_j^{\gamma_j}$ of the consumption good bundle
at the $\varphi$-rigid equilibrium satisfies 
$$\frac{p_k}{\prod_jp_j^{\gamma_j}}=e^{\xi_k}\prod_{j\ne k}\E[e^{\rho_j}|\varphi(\eta)]^{\gamma_j} e^{-\sum_{j}v_j^0\zeta_j}=
\prod_{j\ne k}\E[e^{\rho_j}|\varphi(\eta)]^{\gamma_j} e^{\sum_{j}(L_{kj}-v_j^0)\zeta_j}\,,$$
thus proving Equation \eqref{price-consbundle}. 

From Equations \eqref{FOC-z}, \eqref{1+rk}, \eqref{s-ter}, and \eqref{prices}, we then get 
$$z^0_{jk}
=\frac{A_{jk}s_ke^{-\zeta_k}}{p_j\E[e^{\rho_j}|\varphi(\eta)]}
=A_{jk}e^{-\zeta_k}v_k^{\zeta}\frac{w}{p_j\E[e^{\rho_j}|\varphi(\eta)]}
=e^{-\xi_j}A_{jk}e^{-\zeta_k}v_k^{\zeta}\,,$$
thus proving Equation \eqref{max-z}. 
Similarly, from Equations \eqref{FOC-l}, \eqref{1+rk}, and \eqref{s-ter}, we get 
$$l_k=\frac{\beta_ks_ke^{-\zeta_k}}{w}=\beta_kv_k^{\zeta}e^{-\zeta_k}\,,$$
thus proving Equation \eqref{labor}. Finally, from Equations \eqref{FOC-c}, \eqref{w-E}, and \eqref{s-ter}, we get 
$$c^0_k=\frac{E\gamma_k}{p_k\E[e^{\rho_k}|\varphi(\eta)]}
=\frac{w\gamma_k}{p_k\E[e^{\rho_k}|\varphi(\eta)]\psi(\zeta)}
=\frac{\gamma_ke^{-\xi_k}}{\psi(\zeta)}\,,$$
thus proving Equation \eqref{max-c}. Equations \eqref{actual-y}, \eqref{actual-z}, and \eqref{actual-c} follow by combining each of the three identities in \eqref{actual-quantities} with Equations \eqref{max-y}, \eqref{max-z}, and \eqref{max-c}, respectively.

Observe that 
$$\ba{rcl}\mc E(\eta)
&=&w+\sum_{k\in\mc V}\left(\mc A_k(\eta)-\mc L_k(\eta)\right) \\ 
&=&w+\sum_{k\in\mc V}\left(p_ky_k^{\eta}-\sum_{j\in\mc V}p_jz_{jk}^{\eta}\right)-w\sum_{k\in\mc V}l_k\\
&=&\sum_{k\in\mc V}\left(\frac{we^{\xi_k}e^{\rho_k}v^{\zeta}_ke^{-\xi_k}}{\E[e^{\rho_k}|\varphi(\eta)]}-\sum_{j\in\mc V}\frac{we^{\xi_j}e^{\rho_j}e^{-\xi_j}A_{jk}e^{-\zeta_k}v_k^{\zeta}}{\E[e^{\rho_j}|\varphi(\eta)]}\right)\\
&=&w\sum_{k\in\mc V}v_k^{\zeta}\left(\tau_k-\sum_{j\in\mc V}\tau_jA_{jk}e^{-\zeta_k}\right)\\
&=&\frac{w}{\psi(\zeta)}\gamma'L^{\zeta}(I-e^{-[\zeta]}A')\tau\\
&=&\frac{w}{\psi(\zeta)}\gamma'\tau\,,
\ea$$
where the first identity follows from \eqref{budget}, the second one from \eqref{assets} and \eqref{liabilities}, 
the third one from Equations \eqref{prices}, \eqref{actual-y}, \eqref{actual-z}, and the clearing condition \eqref{labor-clearing} for the labor market, 
the fourth one from the definition \eqref{tau-def} of the information-normalized total shock $\tau_k$, 
 the fifth one from the definition \eqref{centrality-def} of the discounted centrality $v_k^{\zeta}$, 
 and the last one from the definition \eqref{distorted-Leontief} of the discounted Leontief matrix. 
Hence, we have proved Equation \eqref{actual-budget}. 
 
Finally, we have  
$$\sum_{k\in\mc V}p_kc^{\eta}_k
=\sum_{k\in\mc V}\frac{we^{\xi_k}e^{\rho_k}\gamma_ke^{-\xi_k}}{E[e^{\rho_k}|\varphi(\eta)]\psi(\zeta)}=\frac{w}{\psi(\zeta)}\sum_{k\in\mc V}\gamma_k\tau_k
=\mc E(\eta)\,,$$
where the first identity follows from  Equations \eqref{prices} and \eqref{actual-c}, the second one from the definition \eqref{tau-def} of the information-normalized total shock $\tau_k$, and the last one from Equation \eqref{actual-budget}. This proves that not only the conditional expected budget constraint \eqref{expected-budget-constraint} is satisfied, but also the actual budget constraint \eqref{budget-constraint}  is satisfied (with equality).

Equation \eqref{actual-profit} for the actual profit of industry $k$ can be derived as follows: 
\be\label{prof}\ba{rcl}\pi_k(\eta)
&=&\ds\mc A_k(\eta)-(1+r_k\theta_k)\mc L_k(\eta)\\
&=&\ds p_ky^{\eta}_k-e^{\zeta_k}\sum_{j\in\mc V}p_jz_{jk}^\eta-e^{\zeta_k}wl_k\\
&=&\ds w\frac{v_k^{\zeta}e^{\rho_k}}{\E[e^{\rho_k}|\varphi(\eta)]}-w\sum_{j\in\mc V}\frac{v_k^{\zeta} A_{jk}e^{\rho_j}}{\E[e^{\rho_j}|\varphi(\eta)]}-wv_k^{\zeta}\beta_k\\
&=&\ds wv_k^{\zeta}\left(\tau_k-\sum_{j\in\mc V}A_{jk}\tau_j-\beta_k\right)
\\
&=&\ds wv_k^{\zeta}\left(\tau_k-\eps_k\right)
\,,
\ea\ee
where the first identity follows from Equations \eqref{profit}, the second one from \eqref{assets}, \eqref{liabilities}, and \eqref{rkopt}, the third one from \eqref{prices}, \eqref{actual-y}, \eqref{actual-z}, and \eqref{labor}, the forth one from \eqref{tau-def}, and the fifth one from \eqref{epsk}. 

On the other hand, Equation \eqref{actual-utility} for the actual welfare can be derived as follows: 
\eqref{actual-c}
$$\ba{rcl}U(c^{\eta})
&=&\ds\chi\prod_{k\in\mc V}(c_k^{\eta})^{\gamma_k}\\
&=&\ds\prod_{k\in\mc V}\gamma_k^{-\gamma_k} \prod_{k\in\mc V}\frac{\gamma_k^{\gamma_k}e^{\gamma_k(\rho_k-\xi_k)}}{\psi(\zeta)^{\gamma_k}}\\
&=&\ds\frac{\ds e^{\sum_l\sum_k\gamma_jL_{jk}(\eta_k-\zeta_k)}}{\psi(\zeta)^{\sum_{k}\gamma_k}}\\
&=&\ds\frac{\ds e^{\sum_kv^0_k(\eta_k-\zeta_k)}}{\ds\psi(\zeta)}\,,\ea
  $$
  where the first identity follows from \eqref{Cobb-Douglas-2}, the second one from the choice of the normalization constant $\chi$ and \eqref{actual-c}, the third one from \eqref{rho-def} and \eqref{xi-def},  and the last one from \eqref{Bonacich-def} and the constant returns-to-scale assumption \eqref{normalization-2} for the representative household's utility. 

  We are left with proving existence in Theorem \ref{theo:Walrasian}. This follows by choosing $(y^0,z^0,l,c^0,r, p,w)$ according to relations \eqref{max-y}, \eqref{max-z}, \eqref{labor}, \eqref{max-c}, \eqref{rkopt}, and\eqref{prices}. The fact that this is a $\varphi$-rigid Walrasian equilibrium follows by simply inverting previous reasoning.  
\qed

\subsection{Proof of Theorem \ref{prop:no-hulten}}\label{sec:proof-prop-no-hulten}
For every sector $k\in\mc V$, define
$$X_k=\frac{v^0_k}{\psi(\zeta)v^\zeta_k} \sum\limits_{j\in\mc V}  \gamma_j\tau_j-\tau_k\,.$$
It immediately follows from Equations \eqref{pd-welfare} and \eqref{DW} that Hulten's Theorem holds true if and only if \be\label{Xk=0all}X_k=0\,,\qquad\forall k\in\mc V\,.\ee

Observe that 
\be\label{EXk}\ba{rcl}\E[X_k|\varphi(\eta)]
&=&\ds\frac{v^0_k}{\psi(\zeta)v^\zeta_k} \sum\limits_{j\in\mc V}  \gamma_j\E[\tau_j|\varphi(\eta)]-\E[\tau_k|\varphi(\eta)]\\
&=&\ds\frac{v^0_k}{\psi(\zeta)v^\zeta_k} \sum\limits_{j\in\mc V}  \gamma_j-1\\
&=&\ds
\frac{v^0_k}{\psi(\zeta)v^\zeta_k}-1\,,\ea\ee
where the first identity follows from the fact that the Bonacich centrality $v^0_k$ is deterministic and the non-normalized discounted centrality $\psi(\zeta)v_k^\zeta$ is measurable with respect to the public signal $\varphi(\eta)$, the second one follows from Equation \eqref{Etauk=1}, and the third one from the constant returns-to-scale assumption \eqref{normalization-2}. 
Now, observe that, from Lemma \ref{lemma:monotonicity-zeta}(iii), we have that  
$$\psi(\zeta)v^\zeta_k\le \psi(0)v^0_k=v^0_k\,\qquad\forall k\in\mc V\,,$$
and that the above holds true as an equality for every sector $k\in\mc V$ if and only if $\zeta=0$, i.e., the cost of debt is null for every sector. 

If $\zeta\ne0$, then it immediately follows that, for some $k\in\mc V$,  $\E[X_k|\varphi(\eta)]>0$, so that $\P(X_k>0|\varphi(\eta))>0$ and $\P(X_k>0)=\E[\P(X_k>0|\varphi(\eta))]>0$. Hence, in this case, with positive probability, Equation \eqref{no hult} holds true for some sector $k\in\mc V$ and Hulten's Theorem does not hold true. 

On the other hand, if $\zeta=0$, then 
$$X_k=\sum_{j\in\mc V}\gamma_j\tau_j-\tau_k\,,\qquad\E[X_k|\varphi(\eta)]=0\,,\qquad \forall k\in\mc V\,.$$ 
Because of the constant returns-to-scale assumption \eqref{normalization-2}, it follows that, in this case, Equation \eqref{Xk=0all} holds true if and only if 
$$\tau_k=\tau_j\,,\qquad\forall j,k\in\mc V\,.$$
Using Equation \eqref{tau-def}, the above is equivalent to \be\label{rhojrhok}\rho_j=\rho_k+\log\frac{\E[e^{\rho_j}|\varphi(\eta)]}{\E[e^{\rho_k}|\varphi(\eta)]}\,,\qquad \forall j,k\in\mc V\,,\ee
i.e., 
\be\label{rho=rhoo}\rho=Y\1+h(\varphi(\eta))\delta^o\,,\ee
where, for an arbitrary sector  $o$, $Y=\rho_o$ and $h:\mc S\to\R^{\mc V}$ is defined as 
$$(h(\varphi(\eta)))_{j}=\log\frac{\E[e^{\rho_j}|\varphi(\eta)]}{\E[e^{\rho_o}|\varphi(\eta)]}\,,\qquad \forall j\in\mc V\,.$$
Recalling that $\rho=L\eta=(I-A')^{-1}\eta$ and multiplying both sides of equation \eqref{rho=rhoo} by $(I-A')$ from the left, we get that  
\be\label{eta=1}\eta=(I-A')\rho=\rho_o(I-A')\1+(I-A')h(\varphi(\eta))=
Y\beta+g(\varphi(\eta))\,,\ee
where $g(\varphi(\eta))=(I-A')h(\varphi(\eta))$. 
  Hence, if $\zeta=0$ and Hulten's theorem holds true, then Equation \eqref{eta=Y} holds true.  Conversely, if Equation \eqref{eta=Y} holds true for some scalar random variable $Y$ and measurable function $g:\mc S\to\R^{\mc V}$, then $$\rho=L\eta=YL\beta+Lg(\varphi(\eta))=Y\1+h(\varphi(\eta))\,,$$ 
  where $h(\varphi(\eta))=Lg(\varphi(\eta))$, so that $$\log\frac{\E[e^{\rho_j}|\varphi(\eta)]}{\E[e^{\rho_k}|\varphi(\eta)]}=\log\frac{e^{h_j(\varphi(\eta))}\E[e^{Y}|\varphi(\eta)]}{e^{h_k(\varphi(\eta))}\E[e^{Y}|\varphi(\eta)]}=h_j(\varphi(\eta))-h_k(\varphi(\eta))=\rho_j-\rho_k\,,$$
  for every two sectors $j$ and $k\in\mc V$, i.e., Equation \eqref{rhojrhok} holds true. 
  We have thus proved that Hulten's theorem holds true with probability one if and only if Equation \eqref{eta=Y} holds true. 
  
  Finally, notice that, if Hulten's theorem is violated with positive probability,  then there exists some sector $k$ such that $\E[X_k]=0$ and $\P(X_k=0)<1$, so that necessarily $\P(X_k>0)>0$, i.e., \eqref{dlogU>lambda} holds true with positive probability. 
\qed

\subsection{Proof of Theorem \ref{theo:welfare}}\label{sec:proof-theo-welfare}
The claim of Theorem \ref{theo:welfare} follows from Equation \eqref{actual-utility} in Corollary \ref{coro:Walrasian} and the following result. 

\begin{lemma}
For every $\zeta$ in $\R_+^{\mc V}$, 
\be\label{logpsi+v0zeta}\log\psi(\zeta)+\sum_{k\in\mc V}v^0_k\zeta_k\ge0\,.\ee
\end{lemma}
\begin{proof}
For two sectors $j,k\in\mc V$, let 
$$\Omega_{jk}=\bigcup_{l=0}^{+\infty}\left\{\omega=(\omega_0,\omega_1,\ldots,\omega_l)\in\mc V^{l+1}:\,{\omega}_0=j,{\omega}_l=k\right\}\,.$$
For ${\omega}$ in $\Omega_{jk}$, define 
$$\alpha_{\omega}=\prod_{h=1}^{l_{\omega}}A_{{\omega}_h{\omega}_{h-1}}\,,$$
where $l_{\omega}$ is the length of ${\omega}$, i.e., the non-negative integer such that ${\omega}\in\mc V^{l_{\omega}+1}$, and where we adopt the usual convention that an empty product (in this case when $l_{\omega}=0$) is considered equal to $1$. 

From the definition of the distorted Leontief matrix \eqref{distorted-Leontief}, using Taylor expansion, for every $\zeta$ in $\R_+^{\mc V}$ we get
\be\label{expansion-Loentief}
\ba{rcl}\ds L_{jk}^{\zeta} e^{-\zeta_k} &=&
\sum\limits_{l=0}^{+\infty}([e^{-\zeta}]A')^l_{jk}e^{-\zeta_k}\\
&=&\ds\sum\limits_{\omega\in\Omega_{jk}}\alpha_\omega\prod_{h=0}^{l_\omega}e^{-\zeta_{\omega_h}}
\ea\ee
In particular, for $\zeta=0$, we obtain
\be\label{Lalpha}L_{jk}=\sum_{\omega\in\Omega_{jk}}\alpha_\omega\,.\ee
Observe that, thanks to the constant returns-to-scale assumption \eqref{normalization-1} for the production functions, we have that $A'\1=\1-\beta$, so that $\beta=(I-A')\1$, hence 
$L\beta=L(I-A')\1=\1$, i.e., 
$$\sum_{k\in\mc V}L_{jk}\beta_k=1\,,\qquad\forall j\in\mc V\,.$$
Substituting \eqref{Lalpha} in the above we get
\be\label{Lbeta} \sum_{k\in\mc V}\sum_{\omega\in\Omega_{jk}}\beta_k\alpha_\omega=1\,,\qquad\forall j\in\mc V\,.\ee
Using the constant returns-to-scale assumption \eqref{normalization-2} for the representative household's utility  and Equation \eqref{Lbeta}, we get that 
\be\label{prob-dist}\sum_{j\in\mc V}\sum_{k\in\mc V}\sum_{\omega\in\Omega_{jk}}\gamma_j\beta_k\alpha_\omega=1\,.\ee
Substituting now expression \eqref{expansion-Loentief} into \eqref{psi}, using \eqref{prob-dist} and the Jensen inequality, we obtain that 
\be\label{-logpsi}\ba{rcl}-\log\psi(\zeta)
&=&\ds-\log\sum_{j\in\mc V}\sum_{k\in\mc V}\gamma_jL_{jk}^{\zeta}\beta_ke^{-\zeta_k}\\
&=&\ds-\log\sum_{j\in\mc V}\sum_{k\in\mc V}\sum_{\omega\in\Omega_{jk}}\gamma_j\beta_k\alpha_\omega\prod_{h=0}^{l_\omega}e^{-\zeta_{\omega_h}}\\
&\le&\ds\sum_{j\in\mc V}\sum_{k\in\mc V}\sum_{\omega\in\Omega_{jk}}\gamma_j\beta_k\alpha_\omega\sum_{h=0}^{l_\omega}\zeta_{\omega_h}
\,.\ea
\ee
Now, observe that to every pair $(\omega, m)$ where $\omega$ is a walk in $\Omega_{jk}$ and $h$ is an integer in $\{0,\dots , l_{\omega}\}$, we can associate the triple $(i, \omega', \omega'')$ where $i=\omega_h$, $\omega'=\omega_{|\{0,\dots , m\}}$, and $\omega''=\omega_{|\{m,\dots , l_{\omega}\}}$. A direct check shows that the application 
$$(\omega, m)\mapsto (i, \omega', \omega'')$$
is injective and its image is given by the set of triples 
$$\{(i, \omega', \omega'')\,|\, i\in\mc V,\, \omega'\in\Omega_{ji},\, \omega''\in\Omega_{ik}\}$$
This yields the following identity
$$\sum_{\omega\in\Omega_{jk}}\alpha_\omega\sum_{h=0}^{l_\omega}\zeta_{\omega_h}=\sum_{i\in\mc V}\zeta_i\sum_{\omega\in\Omega_{ji}}\alpha_\omega\sum_{q\in\Omega_{ik}}\alpha_q$$
Substituting the above into the right-hand side of \eqref{-logpsi} yields 
$$
\ba{rcl}-\log\psi(\zeta)
&\le&\ds\sum_{j\in\mc V}\sum_{k\in\mc V}\gamma_j\beta_k\sum_{i\in\mc V}\zeta_i\sum_{\omega\in\Omega_{ji}}\alpha_\omega\sum_{q\in\Omega_{ik}}\alpha_q\\ 
&=&\ds\sum_{i\in\mc V}\zeta_i\sum_{j\in\mc V}\gamma_j\sum_{\omega\in\Omega_{ji}}\alpha_\omega\sum_{k\in\mc V}\sum_{q\in\Omega_{ik}}\beta_k\alpha_q\\
&=&\ds\sum_{i\in\mc V}\zeta_i\sum_{j\in\mc V}\gamma_j\sum_{\omega\in\Omega_{ji}}\alpha_\omega\\
&=&\ds\sum_{i\in\mc V}\zeta_i\sum_{j\in\mc V}\gamma_jL_{ji}\\
&=&\ds\sum_{i\in\mc V}\zeta_iv^0_{i}\,,
\ea$$
where the second identity follows from \eqref{Lbeta}, the third one from \eqref{Lalpha}, and the last one from definition \eqref{Bonacich-def}. 
We have thus proved Equation \eqref{logpsi+v0zeta}. \qed\end{proof}

\subsection{Proof of Proposition \ref{prop:lev-effect-qp}}\label{sec:proof-prop-lev-effect-qp}

First, by Proposition \ref{prop-zeta}, if the leverage $\theta_o$ of a sector $o$ increases and the leverage $\theta_k$ of every other sector $k\ne o$ remains the same, then the primary cost of debt $\zeta_o$ of that sector does not decrease, while the primary cost of debt $\zeta_k$ of all other sectors $k\ne o$ is unaltered. 

(i)   
 From Equation \eqref{xi-def} and the fact that the entries of the Leontief matrix $L$ are non-negative, we have that the total cost of debt $\xi_k=\sum_{j\in \mc V}L_{kj}\zeta_j$ of every sector $k$ does not decrease and, for every sector $k\ne o$ that is not a customer of sector $o$, the total cost of debt $\xi_k=\sum_{j\ne o}L_{kj}\zeta_j$ remains unaltered, when $\theta_o$ is increased. 

(ii) By Equation \eqref{labor} of Theorem \ref{theo:Walrasian}, we have that the labor employed by an industry $k\in\mc V$ is $l_k=v_k^{\zeta}\beta_ke^{-\zeta_k}$. If sector $k$ is not a supplier of sector $o$, then Lemma \ref{lemma:monotonicity-zeta}(iv) implies that the discounted centrality $v^{\zeta}_k$ is non-decreasing in $\zeta_o$. Hence, $l_k$ is non-decreasing in $\theta_o$ for every sector $k\ne o$ that is not a supplier of $o$. From the market for labor clearing condition \eqref{labor-clearing}, it follows that the employed labor $l_k$ by at least one sector $k$ among $o$ and its suppliers does not increase when the leverage $\theta_o$ is increased. 


(iii) 
By Equation \eqref{max-c} of Theorem \ref{theo:Walrasian}, we have that the maximal consumption of a sector $k\in\mc V$ at the $\varphi$-rigid equilibrium is given by 
$c^0_k=\gamma_ke^{-\xi_k}/\psi(\zeta)$. Since, by Lemma \ref{lemma:monotonicity-zeta}(ii), the normalization factor $\psi(\zeta)$ is non-increasing in $\zeta$, hence in $\theta$, point (i) of the claim implies that the  maximal consumption of every sector $k$ that is not a customer of sector $o$ does not decrease when $\theta_o$ is increased. 

(iv) The first statement follows from point (i). The second one follows from Equation \eqref{price-consbundle} of Theorem \ref{theo:Walrasian} that the unit price of $k$-the good normalized by the one of the consumption good bundle
at the $\varphi$-rigid equilibrium 
is non-decreasing in $\zeta_o$ (hence in $\theta_o$) if $L_{ko}\ge v^0_o$ and non-increasing in $\zeta_o$ (hence in $\theta_o$) if $L_{ko}\le v^0_o$. \qed

\subsection{Proof of Proposition \ref{prop:convexity}}\label{sec:proof-prop-convexity}

We introduce the function of two variables 
$$f(x, t)=\frac{e^{tx}}{\E[e^{t\eta_o}|\varphi(\eta)]}$$
Using the fact that $\rho=L\eta$, the definition of $\tau_k$ and formula \eqref{epsk}, we have that \eqref{default} is equivalent to the inequality
\be\label{default2} f(\eta_o, L_{ko})<(1-\sum\limits_{j\in\mc V} A_{jk})+\sum\limits_{j\in\mc V} A_{jk}f(\eta_o, L_{jo})\ee
When  $k$ is not reachable from $o$, we have that $L_{jo}=0$ for $j=k$ as well for every $j$ that is a direct supplier of $k$ (i.e., such that $A_{jk}>0$). This implies that expression
\eqref{default2} is false as it boils down to $1<1$. This proves the first item.
For $k\neq o$, we have $L_{ko}=\sum\limits_jA_{jk}L_{jo}$ so that \eqref{default2} can be related to the convexity properties of the function $t\mapsto f(x,t)$. A direct computation shows that
\be\label{derivatives}\frac{\partial f}{\partial t}=f(x, t)(\eta_o-m_1(t))\,,\qquad 
\frac{\partial^2f}{\partial t^2}=f(x, t)((\eta_o-m_1(t))^2-m_2(t))\,.\ee
When condition \eqref{inside} holds true, it follows from the second expression in \eqref{derivatives} that $f(\eta_o, t)$ is concave in $t$ for $t\in [0, \bar t]$. This implies that if $L_k^-<\bar t$, condition \eqref{default2} cannot hold and, thus, sector $k$ is not in default. Instead, under assumption \eqref{notinside}, $f(\eta_o, t)$ is strictly convex in $t$ and thus condition \eqref{default2} is always verified. This proves the result.
\qed

\subsection{Proof of Proposition \ref{prop:convexity2}}\label{sec:proof-prop-convexity2}
Considering that $L_{oo}=1+\sum_jA_{jo}L_{jo}$, the default condition for node $o$ can be written as 
\be\label{default2o} f(\eta_o, 1+\sum_jA_{jo}L_{jo})<(1-\sum\limits_{j\in\mc V} A_{jo})+\sum\limits_{j\in\mc V} A_{jo}f(\eta_o, L_{jo})\ee
When no supplier of $o$ can be reached from $o$ itself, we necessarily have that $A_{jo}L_{jo}=0$ for every node $j$ so that relation \eqref{default2o} coincides with expression \eqref{default-o}. This proves the first claim. 

We now notice that when \eqref{inside2} holds true, it follows from the computations \eqref{derivatives} that $f(\eta_o, t)$ is concave in $t$ for $t\in [0, L_o^-]$ and monotonically non-decreasing so that
$$f(\eta_o, 1+\sum_jA_{jo}L_{jo})\geq
f(\eta_o, \sum_jA_{jo}L_{jo})\geq (1-\sum\limits_{j\in\mc V} A_{jo})+\sum\limits_{j\in\mc V} A_{jo}f(\eta_o, L_{jo})$$
This implies that sector $o$ is not in default.

 Instead, under assumption \eqref{notinside2}, $f(\eta_o, t)$ is strictly convex in $t$ and monotonically non-increasing. This implies that condition \eqref{default2o} is always verified. This completes the proof.

\subsection{Proof for Section \ref{sec:ex-line} }\label{sec:proof-sec7}
On the other hand, Equations \eqref{default}, \eqref{epsk}, and \eqref{normalization-1} imply that, for $k=2,3,4$, industry $k$ defaults if and only if \be\label{default-cond-exp}\tau_k<\eps_k=\sum_{j\in\mc V}A_{jk}\tau_j+\beta_k=\alpha\tau_{k-1}+1-\alpha\,.\ee
In fact, in this case, information-normalized total shocks can be computed explicitly:
 \be\label{tau-4}\tau_1=1\,,\qquad
\tau_k=\frac{\alpha^{k-2}+\lambda}{\lambda}e^{\eta_2\alpha^{k-2}}\,,\quad 
k=2,3,4\,.\ee
Using Equation \eqref{tau-4}, condition \eqref{default-cond-exp} is equivalent to 
\begin{equation}\label{line}
\begin{array}{ll}
e^{\eta_2}(1+\lambda)<\lambda\qquad &k=2\\[5pt]
f_k(\eta_2,\alpha,\lambda)<\lambda(1-\alpha)\qquad &k=3,4\,,
\end{array}
\end{equation}
where 
$$f_k(x,\alpha,\lambda)=(\alpha^{k-2}+\lambda)e^{x\alpha^{k-2}}-(\alpha^{k-2}+\alpha\lambda)e^{x\alpha^{k-3}}\,.$$
Notice that 
$$f_k(0,\alpha,\lambda)=(1-\alpha)\lambda\,,\qquad \lim_{x\to-\infty}f_k(x,\alpha,\lambda)=0\,.$$
Moreover,
$$\frac{\partial f_k}{\partial x}(x,\alpha,\lambda)=\alpha^{k-2}\left(e^{\alpha^{k-2}x}(\alpha^{k-2}+\lambda)-e^{\alpha^{k-3}x}(\alpha^{k-3}+\lambda)\right)\,,$$
so that in particular
$$\frac{\partial f_k}{\partial x}(0,\alpha,\lambda)=\alpha^{k-5}(\alpha-1)<0\,.$$
Furthermore, 
$$\frac{\partial^2 f_k}{\partial x^2}(x,\alpha,\lambda)=\alpha^{2k-5}\left(e^{\alpha^{k-2}x}(\alpha^{k-2}+\lambda)\alpha-e^{\alpha^{k-3}x}(\alpha^{k-3}+\lambda)\right)\,,$$
so that whenever $\frac{\partial f_k}{\partial x}(x,\alpha,\lambda)=0$, we have that 
$$\frac{\partial^2 f_k}{\partial x^2}(x,\alpha,\lambda)=\alpha^{2k-5}e^{\alpha^{k-2}x}(\alpha^{k-2}+\lambda)(\alpha-1)<0\,.$$
This proves that there exists a unique $x^*_k(\alpha,\lambda)<0$ such that $f_k(x_k^*(\alpha,\lambda),\alpha,\lambda)=\lambda(1-\alpha)$ and sector $k$ defaults if and only if $\eta_k<x^*_k(\alpha,\lambda)$. 
In particular, we get $x^*_2(\alpha,\lambda)=\log{\lambda}/(1+\lambda)$ so that sector $o=2$ defaults if and only if $\eta_2<\log\frac{\lambda}{1+\lambda}$. 
In fact, we have that $x^*_k(\alpha,\lambda)\le\log{\lambda}/(1+\lambda)$ for $k\ge2$, so that if $\log{\lambda}/(1+\lambda)\leq \eta_o$, then no sector is default. The exact values $x^*_3(\alpha,\lambda)$ and $x^*_4(\alpha,\lambda)$ can be evaluated numerically as done in Section \ref{sec:ex-line}.

\end{appendix}

\bibliographystyle{abbrv}
\bibliography{Bibliography}

\end{document}